\documentclass[11pt, a4paper]{article}
\usepackage[ruled,linesnumbered,vlined,noend]{algorithm2e}
\SetAlFnt{\footnotesize}  %
\usepackage[font={footnotesize}]{caption}

\usepackage[utf8]{inputenc}
\usepackage[T1]{fontenc}
\usepackage{lmodern}

\usepackage[a4paper, margin=1in]{geometry}
\usepackage[mathcal]{euscript}
\usepackage{amsmath}    %
\usepackage{amssymb}    %
\usepackage{amsthm}     %
\usepackage{bbm}        %
\usepackage{mathtools}  %
\usepackage{physics}    %
\usepackage{braket}     %
\usepackage{dsfont}     %
\usepackage{tikz}
\usepackage{tcolorbox}
\usepackage{relsize}	%
\usepackage{xparse}  	%
\usepackage{MnSymbol}	%
\usepackage{authblk}    %
\usepackage{graphicx}   %
\usepackage{float}      %
\usepackage{subcaption} %
\usepackage{xpatch}
\usepackage{setspace}
\usepackage{enumitem}
\usepackage{float}
\usepackage{etoolbox}
\usepackage{mdframed} 	%
\usepackage{multicol}
\usepackage[dvipsnames]{xcolor} %
\usepackage{thmtools}
\usepackage{doi}
\usepackage[style=alphabetic,   %
            backend=biber,
            natbib=true,
            uniquename=false,
            uniquelist=false,
            maxalphanames=5,
            maxcitenames=10,
            maxbibnames=20,
            hyperref=true]{biblatex}

\usepackage{hyperref}
\usepackage[capitalize, nameinlink]{cleveref}

\hypersetup{
	colorlinks=true,
	linkcolor=MidnightBlue,
	citecolor=ForestGreen,
	urlcolor=RoyalBlue
}

\newtheoremstyle{plainroman} %
{}                          %
{}                          %
{\normalfont}               %
{}                          %
{\bfseries}                 %
{.}                         %
{.5em}                      %
{}                          %

\declaretheorem[name=Theorem, style=plainroman]{theorem}
\declaretheorem[name=Lemma, sibling=theorem, style=plainroman]{lemma}
\declaretheorem[name=Corollary, sibling=theorem, style=plainroman]{cor}
\declaretheorem[name=Proposition, sibling=theorem, style=plainroman]{prop}

\declaretheorem[name=Remark, sibling=theorem, style=plainroman]{remark}

\theoremstyle{definition}

\newtcolorbox{dbox}[1][]{colback=green!5!white, colframe=white, boxrule=0pt, 
    left=0pt, right=0pt, boxsep=3pt,
    fonttitle=\bfseries, #1}

\newtcolorbox{tbox}[1][]{colback=blue!5!white, colframe=white, boxrule=0pt, 
    left=0mm, right=0mm, boxsep=3pt,
    fonttitle=\bfseries, #1}

\DeclareMathOperator*{\argmax}{argmax}
\DeclareMathOperator*{\argmin}{argmin}

\renewcommand{\eqref}[1]{Eq.~(\ref{#1})}
\renewcommand{\tr}{\mathrm{Tr}}

\newcommand{\qaq}{\quad \text{and} \quad}
\newcommand{\inn}{{i \in [n]}}

\newcommand{\iv}{^{-1}}
\newcommand{\ihf}{^{-\frac12}}
\newcommand{\hf}{^{\frac12}}
\newcommand{\dg}{^{\dagger}}
\newcommand{\ra}{\rangle}
\newcommand{\la}{\langle}
\newcommand{\kb}[1]{\ketbra{#1}{#1}}
\newcommand{\ot}{\otimes}
\newcommand{\Lt}{\left}
\newcommand{\Rt}{\right}

\newcommand{\cle}{\mathcal{E}}

\newcommand{\clp}{\mathcal{P}}

\newcommand{\cpab}{\mathrm{CP} (\mathsf A, \mathsf B)}
\newcommand{\cpt}{\mathrm{CPTP} (\mathsf A, \mathsf B)}

\newcommand{\posa}{\mathrm{Pos}(\mathsf{A})}

\newcommand{\posab}{\mathrm{Pos}(\mathsf{A \otimes B})}
\newcommand{\dna}{\mathrm{D}(\mathsf{A})}

\newcommand{\dnab}{\mathrm{D}(\mathsf{A \otimes B})}

\newcommand{\hrm}{\mathrm{Herm}}

\newcommand{\si}{\sigma}
\newcommand{\rr}{\rho}

\newcommand{\Te}{\Theta}

\newcommand{\eps}{\varepsilon}

\newcommand{\pc}{\Pi_\mathcal{C}}

\newcommand{\hr}{\hat{\rho}}

\newcommand{\hrf}{\hat{\rho}_{\mathrm{FPLS}}}

\newcommand{\hrl}{\hat{\rho}_{\mathrm{LS}}}

\newcommand{\rf}{\mathrm{F}}

\newcommand{\sfa}{\mathsf{A}}

\newcommand{\sfb}{\mathsf{B}}

\newcommand{\bbc}{\mathbb{C}}

\newcommand{\dab}{d_{\mathsf{AB}}}

\newcommand{\cab}{\mathcal {C}_{\mathsf A, \mathsf B}}

\newcommand{\fc}{\mathfrak{C}}
\newcommand{\cphi}{\mathfrak{C}(\Phi)}

\newcommand{\frs}{\mathrm{F}(\rho, \sigma)}
\newcommand{\db}{\mathrm{d}_\mathrm{B}}
\newcommand{\dpu}{\mathrm{d}_{\mathrm p}}

\newcommand{\dimab}{d_{\mathsf{A B}}}

\newcommand{\lmin}{\lambda_{\min}}

\usetikzlibrary{calc, shadings, 3d}

\setlist[enumerate]{itemsep=0.5pt, parsep=1pt, topsep=2pt}  %

\crefname{cor}{Corollary}{Corollaries}
\Crefname{cor}{Corollary}{Corollaries}

\newcommand{\inm}{{i \in [m]}}
\newcommand{\me}{{\mathrm{M}_\cle}}
\newcommand{\vo}{\vec o}
\newcommand{\vs}{\vec s}
\newcommand{\mpn}{{\mathrm{M}_{\mathcal{P}_n}}}
\newcommand{\mps}{{\mathrm{M}_{\mathcal{P}_1}}}

\newcommand{\li}{\lambda_{i}}
\newcommand{\hli}{\hat \lambda_{i}}

\newcommand{\lina}{\mathrm{L}(\sfa)}
\newcommand{\linb}{\mathrm{L}(\sfb)}

\newcommand{\smin}{s_{\min}}
\newcommand{\vecm}{\mathrm{vec}}
\newcommand{\na}{\mathrm{N}_{\mathsf{A}}}
\newcommand{\cji}{Choi--Jamiołkowski isomorphism}

\newcommand{\rsin}{\| \rho - \sigma \|_\infty}
\renewcommand{\hrm}{\mathrm{Hr}}

\ifdefined\cleanversion

\else

\fi
\let\oldtableofcontents\tableofcontents
\renewcommand{\tableofcontents}{{\hypersetup{linkcolor=black}
\oldtableofcontents}}

\begin{document}
    
\title{\vspace*{-2cm} Fast and Sure-ious Quantum Process Tomography}      
\date{}

\author[1]{A. Afham}
\author[1]{Sayantan Sen}
\author[1,2]{Marco Tomamichel}

\affil[1]{Centre for Quantum Technologies,
National University of Singapore}

\affil[2]{Department of Electrical and Computer Engineering,
National University of Singapore}

\affil[ ]{
\texttt{afham@nus.edu.sg},
\texttt{sayantan789@gmail.com},
\texttt{marco.tomamichel@nus.edu.sg}
}

\maketitle
\vspace*{-0.5cm}

\begin{abstract}
Quantum process tomography (QPT) protocols are often designed by lifting quantum state tomography (QST) protocols through the Choi--Jamiołkowski isomorphism, at the cost of an extra step enforcing the trace-preserving (TP) constraint.
Existing approaches fall into two classes.
The folklore \textit{partial normalization}, which sandwiches the density estimate between inverse square roots of its input marginal, has a closed form but is used as a heuristic, its only analyses being asymptotic and losing a factor polynomial in the dimension.
Projection-based methods are principled but have no closed forms and scale poorly with dimension.
We resolve this dichotomy by showing that the partial normalization \emph{is} the exact projection onto the Choi states of quantum channels in fidelity (equivalently in Bures or purified distance).
Consequently, \emph{any} QST algorithm with guarantees in fidelity lifts to a QPT algorithm that inherits its sample complexity, through a single closed-form congruence that preserves the Choi rank and is computed in Kraus form without any eigendecomposition of Choi dimension.
We lift two single-copy protocols and prove that the non-adaptive one (\emph{Fidelity Projected Least Squares}) recovers the Choi rank and reaches infidelity $\eps$ from a number of channel uses proportional to $1/\eps$.
Numerically, for channels of low Choi rank, it improves on the state-of-the-art projected least squares in accuracy, and by orders of magnitude in the time and memory of the TP-regularization step and of the returned estimate.
\end{abstract}

\section{Introduction}

Quantum process tomography (QPT) is the task of {explicitly characterizing} an unknown quantum process (channel), and is central to the characterization of quantum devices~\cite{chuang1997prescription, Nielsen2010Quantum}. 
\textit{Ancilla-assisted} QPT~\cite{Leung2003Chois, d2001quantum} reduces QPT to quantum state tomography (QST) through the \cji~\cite{choi1975completely,jamiolkowski1972linear}, which maps completely positive (CP) maps bijectively onto the bipartite positive semidefinite (PSD) cone. However, there is an important caveat.
A QST protocol returns a sufficiently accurate density matrix estimate, but the Choi state of a quantum channel is not \emph{any} arbitrary (bipartite) density matrix, as {its \textit{input marginal} must additionally be maximally mixed}:
\begin{equation}
    \Phi : \lina \to \linb \text{ is a quantum channel} \Longleftrightarrow \fc(\Phi) \in \dnab \text{ and } \tr_\sfb[\cphi] = \frac1{d_\sfa} I_\sfa,
\end{equation} 
where $\fc(\Phi)$ is the normalized Choi state {of} $\Phi$ (see \cref{Sec:Prelims}).
We denote the set of Choi states of quantum channels from $\sfa$ to $\sfb$ as $\cab := \{\rho \in \dnab  : \tr_\sfb(\rho) = I_\sfa/ d_\sfa \} = \fc[\cpt]$, where $\cpt$ denotes the set of all quantum channels from $\sfa$ to $\sfb$.
This \textit{input marginal constraint} corresponds to $\Phi$ being a trace-preserving {(TP)} map {(see \cref{Fig:LiftingDiagram})}. 
Hence any \textit{QST to QPT lift} must include an additional \textit{TP regularization} step to produce the Choi state of a valid channel.

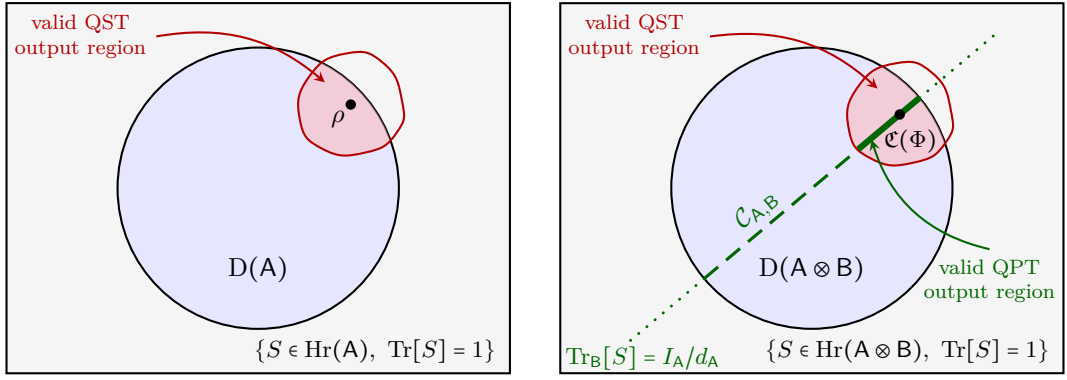
\begin{figure}[h]
    \centering
    \begin{minipage}[c]{0.42\textwidth}
        \centering
        \resizebox{\linewidth}{!}{%
        \begin{tikzpicture}[scale=1.0, >=stealth]
          \draw[thick, fill=gray!8] (-3.4, -2.5) rectangle (3.4, 2.5);
          \node[anchor=south east, font=\footnotesize] at (3.35, -2.5) {$\{S \in \mathrm{Hr}(\sfa),\ \operatorname{Tr}[S] = 1\}$};
          \draw[thick, fill=blue!10] (0, 0) circle (1.9);
          \node[font=\small] at (0, -1.1) {$\mathrm{D}(\mathsf{A})$};
          \coordinate (rho) at (1.25, 1.13);
          \begin{scope}
            \clip (0, 0) circle (1.9);
            \fill[red!30, opacity=0.5] plot[smooth cycle, tension=0.9] coordinates {
              ($(rho)+(0.69, 0.06)$) ($(rho)+(0.44, 0.56)$) ($(rho)+(-0.06, 0.69)$)
              ($(rho)+(-0.56, 0.38)$) ($(rho)+(-0.69, -0.19)$) ($(rho)+(-0.38, -0.60)$)
              ($(rho)+(0.19, -0.65)$) ($(rho)+(0.62, -0.38)$)};
          \end{scope}
          \draw[red!70!black, thick] plot[smooth cycle, tension=0.9] coordinates {
              ($(rho)+(0.69, 0.06)$) ($(rho)+(0.44, 0.56)$) ($(rho)+(-0.06, 0.69)$)
              ($(rho)+(-0.56, 0.38)$) ($(rho)+(-0.69, -0.19)$) ($(rho)+(-0.38, -0.60)$)
              ($(rho)+(0.19, -0.65)$) ($(rho)+(0.62, -0.38)$)};
          \fill[black] (rho) circle (2pt);
          \node[anchor=north east, font=\small, inner sep=1pt] at ($(rho)+(-0.04, -0.04)$) {$\rho$};
          \node[font=\scriptsize, align=center, text=red!70!black] (qst) at (-2.4, 2.05) {valid QST\\output region};
          \draw[->, thick, red!70!black] (qst.east) to[bend left=25] ($(rho)+(-0.35, 0.3)$);
        \end{tikzpicture}}
    \end{minipage}\hspace{0.04\textwidth}%
    \begin{minipage}[c]{0.42\textwidth}
        \centering
        \resizebox{\linewidth}{!}{%
        \begin{tikzpicture}[scale=1.0, >=stealth]
          \draw[thick, fill=gray!8] (-3.4, -2.5) rectangle (3.4, 2.5);
          \node[anchor=south east, font=\footnotesize] at (3.35, -2.5) {$\{S \in \mathrm{Hr}(\sfa \ot \sfb),\ \operatorname{Tr}[S] = 1\}$};
        \draw[thick, fill=blue!10] (0, 0) circle (1.9);
          \node[font=\small] at (0.0, -1.1) {$\mathrm{D}(\mathsf{A} \otimes \mathsf{B})$};
          \coordinate (dir) at (40:1);
          \draw[green!40!black, line width=1.1pt, dash pattern=on 0pt off 4pt, line cap=round] ($-3.2*(dir)$) -- ($-1.9*(dir)$);
          \draw[green!40!black, line width=1.1pt, dash pattern=on 0pt off 4pt, line cap=round] ($1.9*(dir)$) -- ($3.2*(dir)$);
          \draw[green!40!black, very thick, dash pattern=on 7pt off 4pt] ($-1.9*(dir)$) -- ($1.9*(dir)$);
        
          \node[green!40!black, anchor=south, font=\small, inner sep=2pt, rotate=40] at ($-0.75*(dir)$) {$\mathcal{C}_{\mathsf{A},\mathsf{B}}$};
          \coordinate (phi) at ($1.55*(dir)$);
          \begin{scope}
            \clip (0, 0) circle (1.9);
            \fill[red!30, opacity=0.5] plot[smooth cycle, tension=0.9] coordinates {
              ($(phi)+(0.69, 0.06)$) ($(phi)+(0.44, 0.56)$) ($(phi)+(-0.06, 0.69)$)
              ($(phi)+(-0.56, 0.38)$) ($(phi)+(-0.69, -0.19)$) ($(phi)+(-0.38, -0.60)$)
              ($(phi)+(0.19, -0.65)$) ($(phi)+(0.62, -0.38)$)};
          \end{scope}
          \draw[red!70!black, thick] plot[smooth cycle, tension=0.9] coordinates {
              ($(phi)+(0.69, 0.06)$) ($(phi)+(0.44, 0.56)$) ($(phi)+(-0.06, 0.69)$)
              ($(phi)+(-0.56, 0.38)$) ($(phi)+(-0.69, -0.19)$) ($(phi)+(-0.38, -0.60)$)
              ($(phi)+(0.19, -0.65)$) ($(phi)+(0.62, -0.38)$)};
          \draw[green!40!black, line width=2.4pt] ($0.82*(dir)$) -- ($1.9*(dir)$);
          \fill[black] (phi) circle (2pt);
          \node[anchor=north, font=\footnotesize, inner sep=1pt] at ($(phi)+(0.15, -0.16)$) {$\mathfrak{C}(\Phi)$};
          \node[font=\scriptsize, align=center, text=red!70!black] (qst) at (-2.4, 2.05) {valid QST\\output region};
          \draw[->, thick, red!70!black] (qst.east) to[bend left=25] ($(phi)+(-0.3, 0.35)$);
          \node[font=\scriptsize, align=center, text=green!40!black] (qpt) at (2.4, -1.25) {valid QPT \\ output region};
          \draw[->, thick, green!40!black] (qpt.north) to[bend left=30] ($1.25*(dir)+(-0.15, -0.15)$);
          \node[green!40!black, anchor=north east, font=\footnotesize, inner sep=1pt] at (-1.2, -2.1) {$\operatorname{Tr}_{\mathsf{B}}[S] = I_\mathsf{A}/d_\mathsf{A}$};
        \end{tikzpicture}}
    \end{minipage}
    \caption{From state to process tomography. 
    Left, QST: inside the affine space of Hermitian trace-one matrices (rectangle) lies the set of density matrices $\mathrm{D}(\mathsf{A})$ (disc). 
    A QST protocol returns an estimate in a small neighborhood (red outline) of the true state $\rho$ intersected with density matrices (shaded region). 
    Right, QPT on the bipartite space $\mathsf{A} \otimes \mathsf{B}$: the matrices whose $\mathsf{A}$-marginal is maximally mixed form {an affine subspace through the maximally mixed state $I_{\mathsf{AB}}/d_{\mathsf{AB}}$} (green line, dotted outside $\mathrm{D}(\mathsf{A} \otimes \mathsf{B})$), and its intersection with $\mathrm{D}(\mathsf{A} \otimes \mathsf{B})$ (dashed chord) is the set $\mathcal{C}_{\mathsf{A},\mathsf{B}}$ of Choi states of channels. 
    A lifted QPT protocol must regularize a QST output to $\cab$.}
    \label{Fig:LiftingDiagram}
\end{figure}

\paragraph{The TP regularization.}
Running a QST protocol on copies of $\cphi$, treated as a bipartite state on $\sfa \ot \sfb$, returns what we call the \textit{density estimate} $\hr \in \dnab$: accurate by the QST guarantee but blind to the marginal constraint, so in general $\tr_\sfb[\hr] \neq I_\sfa/d_\sfa$ and $\hr$ is the Choi state of a CP map that is not TP. 
Various methods exist for TP regularization, understood as a post-processing step that maps an unconstrained density estimate to a channel, and broadly, they can be classified into two categories which we term \textit{natural} normalizations and \textit{projections}.
(A third family avoids the step altogether by never producing an unconstrained estimate, imposing trace preservation inside the estimator: as a Lagrange multiplier in maximum-likelihood estimation~\cite{fiuravsek2001maximum, jezek2003quantum}, as a linear constraint of a convex program~\cite{shabani2011efficient, Kliesch2019Guaranteed}, through projected gradient descent~\cite{Knee2018Process}, or by a parametrization that is trace preserving by construction~\cite{Ahmed2023GradientDescent}; see \cref{Sec:RelatedWorks}.)
The natural approach involves the following \textit{partial normalization} operation: for $\rr \in \dnab$ with invertible input marginal $\rho_\sfa := \tr_\sfb[\rr]$, consider
\begin{equation} \label{Eq:PartialNormalization}
    \mathrm{N}_{\mathsf{A}}(\rr) := \frac1{d_\sfa} \left(\rho_\sfa \ihf \ot I_\sfb \right) \rr \left(\rho_\sfa \ihf \ot I_\sfb \right).
\end{equation}
The operation $\mathrm{N}_{\mathsf{A}}$ (scaled by $1/d_\sfa$, since we work with trace-normalized Choi states) maps \textit{any} such matrix $\rho$ into $\cab$ and is a non-commutative analogue of constructing a conditional distribution from a joint distribution, $p(b|a) = p(a,b)/p(a)$.

The partial normalization of \cref{Eq:PartialNormalization} is rather folklore, discovered and used repeatedly across quantum information. 
For example, it implements the correspondence between arbitrary bipartite states and channels~\cite{verstraete2003quantumchannels, arrighi2004quantum}; it is the \emph{standard sampler for random channels}~\cite[Eq.~(8)]{bruzda2009random}; in the causally neutral formulation of quantum theory~\cite{leifer2006quantum, Leifer2007Conditional, Leifer2013towards} it \textit{defines} the \textit{conditional state} of $\sfb$ given $\sfa$; and it has long appeared in process tomography itself, as the trace-preservation correction in maximum-likelihood ($R\rho R$) estimation~\cite{fiuravsek2001maximum, jezek2003quantum}, where it appears in the Lagrange multiplier associated with the TP-constraints, as the normalization step of the two-stage protocols of~\cite{Xiao2022Twostage, Xiao2023AAPT, Xiao2025TwoStage}, and as a closed-form heuristic approximation to the intractable Frobenius projection onto $\cab$ in~\cite{barbera2025boosting}.
In each case it enters either as a definition or as a heuristic. 
Without quantitative guarantees (such as optimality), closeness of the density estimate $\hr$ to $\cphi$ does not provide guarantees on the closeness of channel estimate $\mathrm{N}_{\mathsf{A}}(\hr)$ to $\cphi$.
The only error analysis of the operation in process tomography~\cite{Xiao2023AAPT, Xiao2025TwoStage}, that we are aware of, is an asymptotic bound on the expected Frobenius error of a specific linear-regression based QPT protocol, in which the correction step costs a factor of order $d_\sfa^{3/2}$, and its authors note that the step is not the nearest valid point.
By contrast, the guarantee of this article (\cref{Thm:Main}) holds at every sample size, for any density estimate, in fidelity, and the correction step costs a dimension-independent constant factor.

Projection-based regularizations follow a more principled approach as, by definition, they carry optimality.
\textcite{surawy2022projected} projected the {Projected Least Squares (PLS)} density estimate~\cite{guta2020fast} onto $\cab$ in Hilbert--Schmidt distance and \textcite{Mele2025Optimal} in diamond distance~\cite{kitaev1997quantum} via a semidefinite program (SDP).\footnote{\cite{Mele2025Optimal} studies sample complexity, for which this projection is only a final classical-postprocessing whose runtime is unimportant.}
We refer to the QST estimator of~\cite{guta2020fast} as PLS-QST, to the lifted QPT protocol of~\cite{surawy2022projected} as PLS-QPT, and to its TP regularization step, the \textit{hyperplane intersection projection}, as HIP; the resulting channel estimates are labeled PLS in the figures.
Thus a rigorous inheritance of guarantees requires the TP regularization to be \textit{stable}, that is, $\mathrm{d}(\mathrm{R}(\hr), \cphi) \leq c \mathrm{d}(\hr, \cphi)$ for every density estimate $\hr$ and some distance $\mathrm{d}$, with $c$ independent of the sample size and the dimension.
Optimality provides this at once, via triangle inequality, which is why projection-based regularizations transfer guarantees from QST to QPT algorithm.
A non-optimal regularization needs a separate stability analysis.
However, the projection step comes with a significant computational bottleneck: both projections are expensive, owing to the geometry of the set projected onto; \cref{Fig:Headline}~(left) compares their runtime with our \textit{fidelity projection} (discussed next).

\paragraph{A dichotomy and its resolution.}
Projections thus carry guarantees but are expensive, while the partial normalization is cheap but has lacked guarantees.
Ideally, one would have the \emph{best of both worlds}, a regularization that is both optimal and cheap to compute, and this is exactly what we achieve in this work.
Using a recent result on closed-forms for projections with respect to fidelity~\cite{afham2026projections}, this article strikes a unifying note: beyond a convenient way of restoring trace preservation, the partial normalization is the \textit{exact} projection onto $\cab$ with respect to fidelity (equivalently to Bures and purified distances). 
For $\rho \in \dnab$ with $\rho_\sfa \equiv \tr_\sfb[\rho] > 0$, define $\pc(\rho) := \argmax_{\sigma \in \cab} \rf(\rho, \sigma)$. Then
\begin{equation} \label{Eq:FidelityProjectionDef}
     \pc(\rho) = \na(\rho),
\end{equation}
where $\rf(\rho, \sigma) := \|\rho\hf \si\hf\|_1$ is the \textit{(root) fidelity}, and the equality is the main result of this article (\cref{Thm:Main}).
Any fidelity-based QST algorithm therefore lifts to a QPT algorithm at the cost of running it at accuracy $\eps/4$; the factor $4$ is a worst-case constant, and the overhead in our numerical experiments is close to $1$ (see \cref{Sec:Lifting}  and~\cref{Fig:Contraction}: in practice the lift essentially costs no extra accuracy).

\begin{figure}[t]
    \centering
    \includegraphics[width=\textwidth]{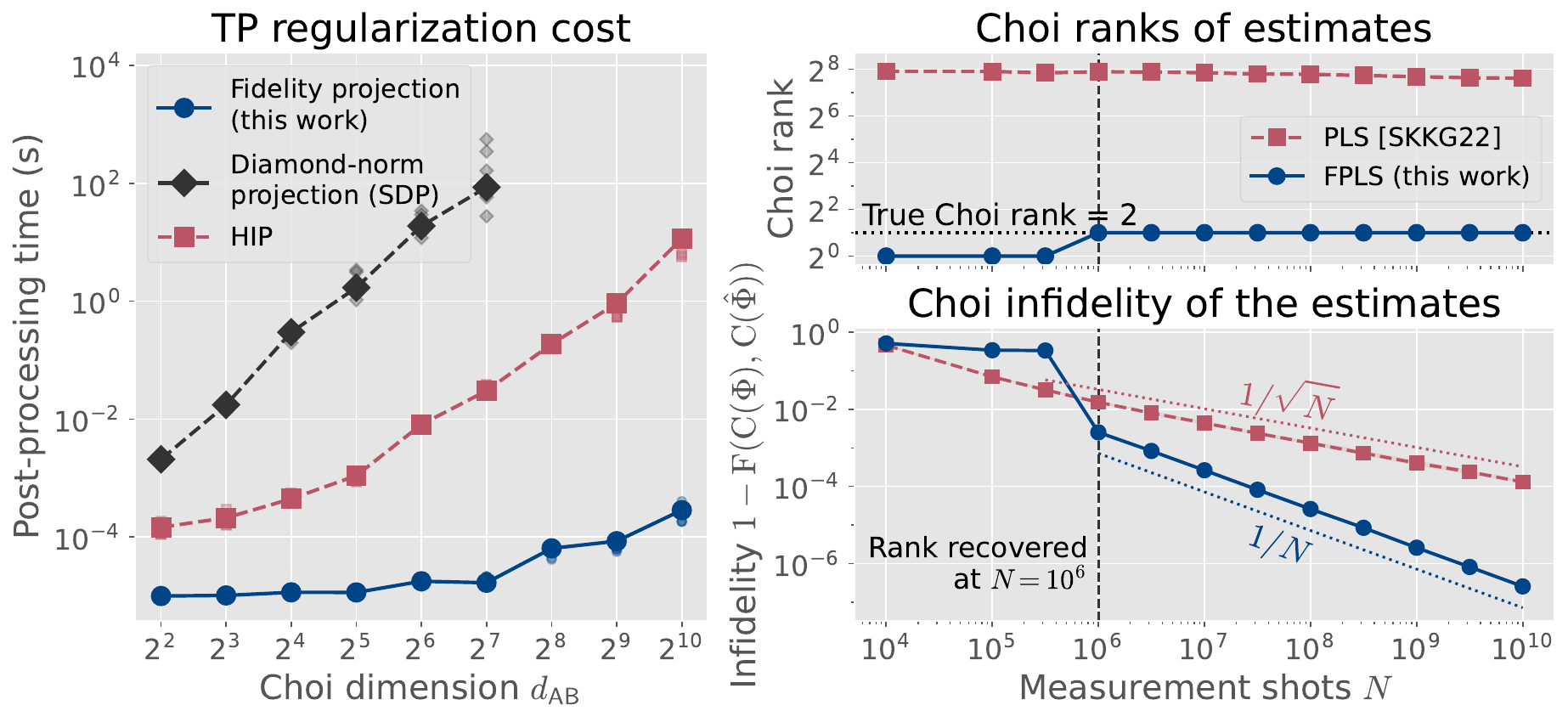}
    \caption{\textbf{(Left)} Runtime of different TP regularization methods against Choi dimension: the iterative \textit{Hyperplane Intersection Projection} (``HIP'')~\cite{surawy2022projected}, our fidelity projection~(\cref{Thm:FactorizedImplementation}), and the diamond-distance projection SDP~\cite{Mele2025Optimal}.
    All three regularize the \emph{same} instance, the QST output obtained from a Haar-random isometry.
    The SDP curve is the solver's own solve time, which understates its cost (\cref{App:HeadlineLadder}).
    \textbf{(Right)} Recovered Choi rank and Choi infidelity with the true channel, for Fidelity Projected Least Squares (FPLS, this work) against PLS~\cite{surawy2022projected}, both applied to the \emph{same} least-squares estimate of a rank-$2$ channel at $\dimab = 2^8$.
    The dashed line at $N = 10^6$ marks the (empirical) onset of \textit{exact rank recovery}, beyond which every FPLS trial returns the true rank $2$; the guaranteed onset of~\cref{Cor:BernsteinRank} is $3 \times 10^6$.
    PLS returns rank $178$--$240$ of $256$ throughout.
    The dotted guides carry the two \emph{predicted} rates $1/N$ and $1/\sqrt{N}$.
    \Cref{App:HeadlineImplementation} details the construction of both panels in full.}
    \label{Fig:Headline}
\end{figure}

Thus we have a TP regularization method that is (i) a projection, (ii) identical to the commonly used partial normalization, and (iii) fast to compute, by its closed form (see~\cref{Fig:Headline}). 
It is also exact: HIP and the SDP return their projection only up to a stopping tolerance, so their runtime depends on how tight that tolerance is set, as for any numerical approximation scheme.
We run HIP at its authors' default settings throughout the cost comparisons (\cref{App:HeadlineLadder}); the closed form has no such parameter and is evaluated to machine precision at a fixed cost.
We emphasize that it is not the formula, but its \textit{optimality}---in terms of it being a projection---that is new, which is what allows for transfer of guarantees from QST to QPT. 
A projection satisfies $\db(\hr, \pc[\hr]) \leq \db(\hr, \cphi)$, so by the triangle inequality the regularized estimate is at most twice as far from $\cphi$ as the non-regularized estimate~(see \cref{Sec:Lifting}).

\paragraph{Kraus- and Stinespring-representation perspectives.} 
Note that $\pc$ corresponds to the \textit{Choi--Bures} projection\footnote{The names record the representation and the metric: \textit{Choi--Bures} is a projection of Choi states in Bures distance, and \textit{Stinespring--Frobenius}, below, one of Stinespring operators in Frobenius distance.} of a CP map onto the set of {completely positive trace-preserving (CPTP)} maps. 
The projection also has compelling interpretations in terms of Kraus and Stinespring representations~\cite{kraus1971general,Stinespring1955Positive}.
A \emph{Kraus representation} $\{K_i\}$ of the CP map becomes $\{K_i R\ihf\}$ with $R := \sum_i K_i^\dagger K_i$, and a \emph{Stinespring operator} $K$ is replaced by its polar factor $V = K|K|\iv$, which is itself the Frobenius projection of $K$ onto the set of isometries~\cite[Theorem IX.7.2]{bhatiamatrixanalysis}. 
The TP regularization is thus both the \textit{normalization} of the Kraus operators of the density estimate and the \textit{Stinespring--Frobenius} projection onto the set of isometries (see also~\cite[Section 2.3]{afham2026projections}).

\paragraph{Rank preservation.}
Being a congruence by an invertible operator, the fidelity projection satisfies $\rank(\pc[\hr]) = \rank(\hr)$, and to the best of our knowledge it is the only TP regularization that achieves this (provided $\hr_\sfa$ is invertible).
The Choi rank is the minimal number of Kraus operators, and the channels of typical interest in quantum information theory have low Choi rank. 
Crucially, rank preservation lets a rank-selecting QST algorithm lift to a QPT algorithm returning a channel with correspondingly few Kraus operators, storable and applicable as $r$ operators instead of a $\dimab \times \dimab$ matrix, whereas a \textit{rank-inflating} regularizer destroys that structure in the last step.
Our non-adaptive protocol, \textit{Fidelity Projected Least Squares} (FPLS), is a rank-selecting QST algorithm lifted through the rank-preserving fidelity projection, and it returns the Choi rank exactly: in~\cref{Fig:Headline}~(top right), on a channel of Choi rank $2$ at $\dimab = 2^8$, it returns \emph{exactly} rank $2$ at every shot count above its onset, while PLS returns rank $178$--$240$ of $256$ throughout; rank \textit{selection} by the estimator is useless unless the regularizer \textit{preserves} it, and~\cref{Sec:Numerics} explains why HIP cannot preserve rank in general. 
On every instance reported here, HIP returns an estimate with high Choi rank. 

\begin{figure}[t]
    \centering
    \includegraphics[width=\textwidth]{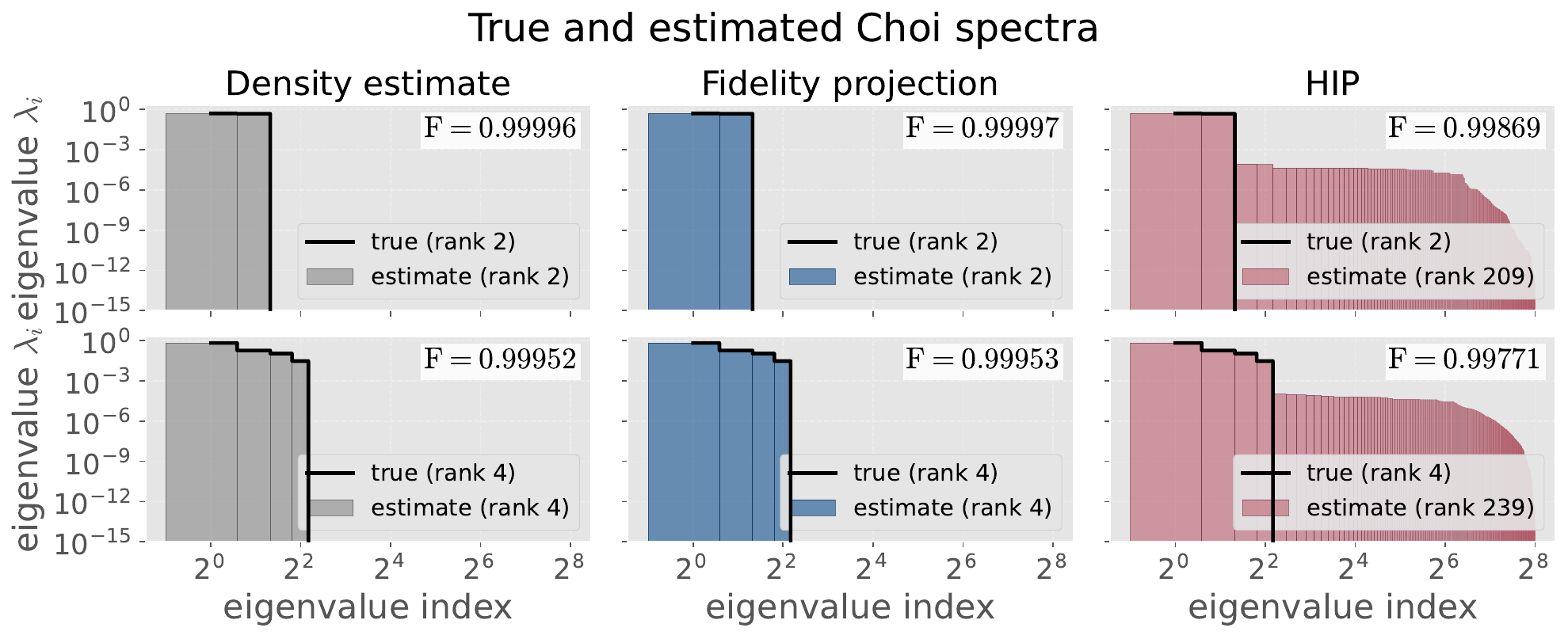}
    \caption{Rank preservation in practice, at $\dimab = 2^8$. 
    Top row: the rank-$2$ channel of~\cref{Fig:Headline}; bottom row: independent amplitude damping on two of the four qubits, Choi rank $4$; both at $N = 10^8$. 
    Left: ordered spectrum of the density estimate, \cref{Alg:SparseDensityEstimate} at $\tau = \beta_N/2$, which corresponds to a CP but not TP map. Middle and right: the channel estimates obtained by TP regularizing this \emph{same} density estimate with fidelity projection and with HIP. 
    The black line marks the true eigenvalues, the bars those of the estimate, and the inset $\rf$ is the fidelity of the estimate with the true Choi state.}
    \label{Fig:SpectrumMatching}
\end{figure}

\paragraph{Related Works.} \label{Sec:RelatedWorks}
Process tomography protocols are often characterized in terms of their access to the unknown quantum channels and the sample complexity. 
The unknown channels can be accessed either jointly, in the \textit{multi-copy} setting~\cite{Mele2025Optimal, chen2026quantum}, or one at a time, adaptively or not, in the \textit{single-copy} setting~\cite{surawy2022projected} we work in; matching bounds are now known in some of these settings~\cite{chen2026quantum, oufkir2023sample}.

Beyond projection-based regularization, QPT has been approached by gradient and maximum-likelihood methods~\cite{Knee2018Process, fiuravsek2001maximum, quiroga2023using}, Bayesian inference~\cite{Granade2016}, two-stage normalizations~\cite{barbera2025boosting, Xiao2023AAPT, Xiao2025TwoStage}, and compressed sensing~\cite{shabani2011efficient, roth2018recovering, Kliesch2019Guaranteed}, which compresses the measurement stage under a rank promise and recovers by convex programming; our contribution is orthogonal, making the \emph{post-processing} exact for any fidelity-based front end, with rank exploited a posteriori rather than promised.
The closest rigorous and practical alternative to PLS-QPT is~\cite{Kliesch2019Guaranteed}: single-copy, non-adaptive, with diamond-norm recovery from $O(r d_\sfa^2 \log d_\sfa)$ settings.
It differs in architecture, as physicality is imposed as a feasibility constraint inside an SDP, whereas~\cref{Thm:Main} lifts \emph{any} QST algorithm. 
Moreover, its guarantees require $4$-designs rather than the local Pauli measurements used here.
On the state side, truncating the spectrum of a least-squares estimate has been studied in~\textcite{butucea2015spectral}.
Our density estimate is of this kind; what we add is its analysis in fidelity, where recovering the rank exactly improves the rate from $1/\eps^2$ to $1/\eps$, and its composition with a TP regularization that preserves the rank.
Finally, if additional {structure is} assumed on the channel, computationally and statistically efficient algorithms have been designed: such as Pauli channels~\cite{flammia2020efficient,chen2022quantum,rouze2023efficient,fawzi2025lower,chen2025efficient}, unitary channels~\cite{acin2001optimal, bagan2004quantum,hayashi2006parallel,yang2020optimal,haah2023query}, and isometry channels~\cite{yoshida2026quantum}.

\paragraph{Novelty relative to prior uses of partial normalization.} \label{Sec:NoveltyPN}
The partial-normalization formula~\cref{Eq:PartialNormalization} is folklore in QPT: it appears in~\cite{jezek2003quantum, Xiao2022Twostage, Xiao2023AAPT, Xiao2025TwoStage, barbera2025boosting}, and in~\cite{barbera2025boosting}, the closest work to ours, it is reached as a heuristic approximation to the Frobenius projection onto $\cab$ by solving an optimization~\cite[Eq.~(12)]{barbera2025boosting}. 
In our companion paper, we show that is in fact the exact projection in Bures distance~\cite[Sec.~2.3]{afham2026projections}.
None of these previous works proves that projection property, which is essential for the inheritance of sample complexity we show in~\cref{Thm:Main}.
What this article contributes is the consequence for QPT: the generic QST-to-QPT lifting theorem~(\cref{Thm:Main}), the finite-sample analysis of FPLS-QPT, the rank-factorized Kraus implementation of the correction~(\cref{Thm:FactorizedImplementation}), and the empirical comparison of~\cref{Sec:Numerics}.
The speed and memory advantages we report come from this exactness in the fidelity geometry together with the factorized implementation, not from the formula alone.

\begin{table}[t]
\centering
\small
\setlength{\tabcolsep}{4pt}
\renewcommand{\arraystretch}{1.15}
\begin{tabular}{llcc|ccc|ccc}
\hline
 & & \multicolumn{2}{c|}{Infidelity $1 - \rf$} & \multicolumn{3}{c|}{TP regularization time} & \multicolumn{3}{c}{Storage of the Estimate} \\
$\dimab$ & true rank & PLS & FPLS & PLS & FPLS & ratio & PLS & FPLS & ratio \\
\hline
    $2^{6}$ & $2$ & {$8.1{\times}10^{-4}$} & {$\mathbf{4.2{\times}10^{-6}}$} & 3.6\,ms & \textbf{0.02\,ms} & $190\times$ & 66\,kB & \textbf{2\,kB} & $32\times$ \\
    $2^{10}$ & $2$ & {$4.1{\times}10^{-3}$} & {$\mathbf{1.6{\times}10^{-4}}$} & 6.9\,s & \textbf{0.29\,ms} & $2.4\!\times\!10^{4}$ & 16.8\,MB & \textbf{33\,kB} & $512\times$ \\
    $2^{6}$ & $\log_2 \dimab$ & {$1.0{\times}10^{-3}$} & {$\mathbf{4.2{\times}10^{-5}}$} & 3.3\,ms & \textbf{0.03\,ms} & {$120\times$} & 66\,kB & \textbf{6\,kB} & $11\times$ \\
    $2^{10}$ & $\log_2 \dimab$ & {$1.2{\times}10^{-2}$} & {$\mathbf{4.6{\times}10^{-3}}$} & 3.7\,s & \textbf{1.02\,ms} & $3.6\!\times\!10^{3}$ & 16.8\,MB & \textbf{164\,kB} & $102\times$ \\
\hline
\end{tabular}
\caption{PLS~\cite{surawy2022projected} against FPLS (this work) on the same $N = 10^8$ local Pauli shots, medians over three trials: Choi infidelity of the channel estimate, wall-clock of the TP regularization step (HIP against the factorized projection of~\cref{Thm:FactorizedImplementation}), and storage of the returned estimate (dense matrix against $r$ Kraus operators, in double precision).
The rank-$2$ channel is that of~\cref{Fig:Headline}; the other is a random channel of Choi rank $\log_2 \dimab$~\cite{bruzda2009random}. 
PLS-QPT is run with its authors' code at their default settings (\cref{App:HeadlineLadder}).
FPLS returns the true Choi rank in every case, PLS-QPT $63$ of $64$ and at least $1016$ of $1024$.}
\label{Tab:HeadToHead}
\end{table}

\paragraph{Organization.}
In \Cref{Sec:Prelims}, we define the notation. Then in \Cref{Sec:Lifting}, we develop the main results: the closed-form fidelity projection and its optimality, the lifting theorem with the resulting non-adaptive and adaptive QPT protocols and their sample complexities, and the choice of metric.
Later in \Cref{Sec:Numerics}, we discuss the numerical experiments, benchmarked throughout against PLS-QPT~\cite{surawy2022projected}.
The appendices review the non-adaptive QST algorithm we lift and prove relevant guarantees, describe the instances used in numerical experiments, and present the remaining proofs.

\section{Notation and preliminaries} \label{Sec:Prelims}
We now introduce our formal notations and collect some useful mathematical results. 
We denote finite-dimensional complex Hilbert spaces {by} $\sfa, \sfb${,} with dimensions denoted as $\dim(\sfa) \equiv d_\sfa, \dim(\sfb) \equiv d_\sfb$, respectively. We will also use the shorthand $\dimab \equiv d_\sfa d_\sfb$.
$[n] = \{1, 2, \ldots, n\}$ denotes the index set.
The {sets of linear operators}, Hermitian, positive semidefinite, and density matrices on $\sfa$ will be denoted as $\mathrm{L}(\sfa),{\hrm(\sfa)}, \posa, \dna$. 
We denote the set of completely positive maps and quantum channels from $\sfa$ to $\sfb$ as $\cpab$ and $\cpt$, respectively.
$I_\sfa$ denotes the identity matrix on $\sfa$ and $\mathrm{Id}_\sfa$ denotes the identity channel on $\mathrm{L}(\sfa)$. 

For PSD matrices $P,Q \in \posa$, the \textit{(square-root) fidelity} between them is defined as $\rf(P, Q) := \| P \hf Q \hf \|_1$~\cite{uhlmann1976transition, jozsa1994fidelity}.
Fidelity is related to two important distances in quantum information: the \emph{Bures distance} and the \emph{purified distance}~\cite{bures1969extension, tomamichel2013thesis}.
For PSD matrices $P,Q \in \posa$ and density matrices $\rho, \sigma \in \dna$, they are defined as
\begin{equation}
    \mathrm{d_B}(P,Q) := \sqrt{\tr[P+Q] - 2 \rf(P,Q) }, \qaq \dpu(\rho, \si) = \sqrt{1 - \rf(\rho, \si)^2}.
\end{equation}
For $A \in \mathrm{L}(\sfa)$ the \textit{trace norm} is $\|A\|_1 := \tr[\sqrt{A^\dagger A}]$ and the \textit{Frobenius norm} is $\|A\|_2 := \sqrt{\tr[A^\dagger A]}$, giving the \textit{trace distance} $\frac12 \|\rho - \sigma\|_1$ and the \textit{Frobenius distance} $\|\rho - \sigma\|_2$ between states.
The Fuchs--van de Graaf inequalities~\cite{fuchs1999cryptographic} relate the Bures and purified distances to the trace distance:
\begin{equation} \label{Eq:FuchsVan}
    \frac12\db^2(\rr, \si) := 1 - \rf(\rr, \si) \leq \frac12 \|\rho - \si \|_1 \leq \sqrt{1 - \rf(\rr, \si)^2} =: \dpu(\rr, \si),
\end{equation}
for all $\rho, \si \in \dna$.
The \textit{(normalized) Choi representation} of a CP map $\Te \in \cpab$ is defined as $\fc(\Te) := (\mathrm{Id}_\sfa \ot \Te)(\kb{\omega})$, where $|\omega \ra = \frac{1}{\sqrt{d_\sfa}} \sum_{i}^{d_\sfa} |i,i \ra $ is the canonical maximally entangled state. 
The Choi representation bijectively maps $\cpab$ to $\posab$, allowing us to write $\fc[\cpab] = \posab$.
We write $\tr_\sfb$ for the partial trace over $\sfb$, and abbreviate the marginals of a bipartite operator $P \in \posab$ by subscripts: $P_\sfa := \tr_\sfb[P]$ and $P_\sfb := \tr_\sfa[P]$.
Throughout, $\rho \in \cab$ denotes the Choi state of the unknown channel and $\hr \in \dnab$ an estimate of it, say the output of a QST algorithm; we call $\hr$ the \textit{density estimate} and stress that it is an arbitrary bipartite state, which need not lie in $\cab$, its marginal $\hr_\sfa$ being in general different from $I_\sfa/d_\sfa$.
Specific estimates will carry subscripts, such as the least-squares estimate $\hrl$ of~\cref{Eq:LSEstimate}.

\section{Lifting QST algorithms to QPT algorithms} \label{Sec:Lifting}
In this section, we present our approach of lifting QST algorithms to QPT algorithms. We first formally identify the partial-normalization based TP regularization as a fidelity projection, and discuss its implications for QPT. 
We then lift two single-copy QST algorithms, the non-adaptive \textit{Projected Least Squares} QST algorithm by~\textcite{guta2020fast} and the adaptive QST algorithm by~\textcite{chen2023does}. 

Our lifting approach can be summarized as a composition of the partial normalization map and the QST algorithm: $\mathcal{A}_{\mathrm{QPT}} := \mathrm{N}_{\sfa} \circ \mathcal{A}_{\mathrm{QST}}$.
The lifted QPT protocol inherits the sample complexity of \textit{any} QST protocol whose guarantees are given in fidelity, up to an overhead (in the number of samples) by a constant factor.
The following theorem identifies the partial normalization as the fidelity projection onto $\cab$, bounds the loss of projecting, and lifts QST guarantees to QPT.
\begin{tbox}
    \begin{theorem}[TP regularization via fidelity projection] \label{Thm:Main}
    Let $\rho \in \dnab$ with $\rho_\sfa  > 0$. Define
    \begin{equation} \label{Eq:ClosedFormProj}
        \na(\rho) := \frac1{d_\sfa} \Lt(\rho_\sfa \ihf \ot I_\sfb\Rt) \rho \Lt(\rho_\sfa \ihf \ot I_\sfb\Rt).
    \end{equation}
    \begin{enumerate}
        \item \emph{Projection.} $\na(\rho)$ is the unique maximizer of $\rf(\rho, \sigma)$ over $\sigma \in \cab$, and equivalently the unique minimizer of $\db(\rho, \sigma)$ and of $\dpu(\rho, \sigma)$. We call it the \emph{fidelity projection} of $\rho$ onto $\cab$ and write $\pc(\rho) := \na(\rho)$.
        \item \emph{Stability.} Let $\rho \in \cab$ be the Choi state of a channel and let $\hr \in \dnab$ be bipartite state with $\hr_\sfa > 0$. Then
        \begin{equation} \label{Eq:ProjLoss}
            \db(\rho, \pc(\hr)) \leq 2\, \db(\rho, \hr) \qaq 1 - \rf(\rho, \pc(\hr)) \leq 4 \left(1 - \rf(\rho, \hr)\right).
        \end{equation}
        \item \emph{Lifting.} Let $\eps, \delta \in (0, 1)$, and let $\mathcal{A}_\mathrm{QST}$ be a QST algorithm that, given $N(\eps)$ copies of an unknown state, returns an estimate of fidelity at least $1 - \eps$ with that state, with probability at least $1 - \delta$.
        Let $\Phi \in \cpt$ be an unknown channel, and let $\hr$ be the output of $\mathcal{A}_\mathrm{QST}$ run at accuracy $\eps/4$ on $N(\eps/4)$ copies of its Choi state $\cphi$.
        Then, with probability at least $1 - \delta$, the channel $\hat\Phi$ with Choi state $\fc(\hat\Phi) := \pc(\hr)$ satisfies
        \begin{equation}
            \rf(\cphi, \fc(\hat\Phi)) \geq 1 - \eps .
        \end{equation}
        That is, $\mathcal{A}_\mathrm{QPT} := \pc \circ \mathcal{A}_\mathrm{QST}$ is a QPT algorithm of accuracy $\eps$ and confidence $1 - \delta$ that uses $N(\eps/4)$ copies.
    \end{enumerate}
    \end{theorem}
\end{tbox}
\begin{proof}
    (1) Note that the marginal is $\tr_\sfb[\na(\rho)] = \frac1{d_\sfa}\rho_\sfa\ihf \rho_\sfa \rho_\sfa\ihf = I_\sfa/d_\sfa$, {which} implies that $\na(\rho) \in \cab$. The optimality guarantee follows from the fixed-marginal specialization of~\cite[Corollary 14]{afham2026projections}. Moreover, since $\rf$, $\db$ and $\dpu$ are monotone transforms of one another on states, the three optimizations have the same solution.
    (2) Since $\rho \in \cab$ is feasible, the optimality of the projection gives $\db(\hr, \pc(\hr)) \leq \db(\hr, \rho)$. 
    Applying the triangle inequality, we have $\db(\rho, \pc(\hr)) \leq \db(\rho, \hr) + \db(\hr, \pc(\hr)) \leq 2\,\db(\rho, \hr)$. 
    Using the fact that $\db^2 = 2(1 - \rf)$, we conclude that $1 - \rf(\rho, \pc(\hr)) \leq 4 \left(1 - \rf(\rho, \hr)\right)$.
    (3) Running $\mathcal{A}_\mathrm{QST}$ at accuracy $\eps/4$ ensures that $1 - \rf(\rho, \hr) \leq \eps/4$. 
    Applying (2) gives the result.
\end{proof}

The hypothesis $\rho_\sfa > 0$ constrains only the marginal and not the bipartite state, which may be rank deficient; an estimate whose marginal is rank deficient can moreover always be extended, though not uniquely, to one with a full-rank marginal, which we do not pursue here for the reasons given in~\cref{Rem:Singular}.
Note that the overhead of lifting QST protocols to QPT protocols is a change of \textit{accuracy parameter}, not a fixed multiple of the sample count. In particular, the cost of lifting is $N(\eps/4)/N(\eps)$, a factor $4$ for an algorithm scaling as $1/\eps$ and $16$ for one scaling as $1/\eps^2$.
The constants $2$ and $4$ are for the worst case; in practice the projection is far less lossy, since by {the} data processing inequality {(DPI)}, we have $\rf(\rho_\sfa, \hat \rho_\sfa) \geq \rf(\rho, \hat \rho)$, and $\rho_\sfa = I_\sfa/d_\sfa$ for $\rho \in \cab$, so an accurate estimate already has a marginal close to $I_\sfa/d_\sfa$.
In practice the lift therefore costs essentially no extra accuracy: over $870$ tomography estimates of channels from eight families, with $\dimab$ from $2^2$ to $2^{10}$, the Bures distance to the true Choi state after the projection is at most $1\%$ above the distance before it, against the worst case of a factor $2$, and for low-rank channels it is often well below (\cref{App:ContractionNumerics}, \cref{Fig:Contraction}; on the calibrated instances of~\cref{App:SyntheticInstances} it never increased).
The reason is that the worst case of~\cref{Eq:ProjLoss} requires states to be far from $\cab$, whereas a tomography estimate is close to the Choi state of the true channel, and hence to $\cab$, by construction.

\begin{remark}[Invertibility of the marginal] \label{Rem:Singular}
    If (the density estimate) $\hat \rho_\sfa$ is singular, the inverse in~\cref{Eq:ClosedFormProj} is taken on the support $\operatorname{supp}(\hat \rho_\sfa)$; the result is a CP map that is trace preserving only on $\operatorname{supp}(\hr_\sfa)$ and whose extension to a channel is not unique, so all statements in this section assume $\hat\rho_\sfa > 0$ (one explicit extension is given in~\cref{App:Invertibility}).
    For tomography protocols, however, the assumption holds automatically at reasonable accuracies: the marginal of an accurate estimate is close to $I_\sfa/d_\sfa$ by the bound above, and any $\hr$ with $\rf(\rho, \hr) > 1 - 1/(2 d_\sfa)$ has an invertible marginal (see \cref{Lem:AutoInvertible}).
    Consequently, for sufficiently large $N$, the estimate has an invertible $\sfa$-marginal with high probability.
\end{remark}
\paragraph{On the choice of metric.}
Tomographically, any loss function serves; however operational interpretations differ. For concreteness, we consider three metrics.
The \textbf{Frobenius} projection onto $\cab$ is non-expansive, which is what alternating schemes such as HIP rely on. However, it has several drawbacks: it has no closed form, its iterations scale poorly with dimension (see \cref{Fig:Headline}, left), and the distance violates the data-processing inequality~\cite{ozawa2000entanglement}.
The \textbf{diamond} distance is the operational worst case and the right figure of merit for a channel used inside a larger circuit, but its projection must be computed via a semidefinite program, is not unique, and is computationally prohibitive (see \cref{Fig:Headline}, left). 
In contrast, the \textbf{Bures/fidelity} projection recovers the {widely used} partial normalization (see~\cref{Eq:PartialNormalization} and \cref{Fig:Headline} (left)). Additionally, {of the three, it is the only one whose projection has a closed form}, is inexpensive to compute in \textit{rank-factorized form} (\cref{Thm:FactorizedImplementation}), and {is} rank preserving; the Choi fidelity is moreover related to the average and entanglement fidelities of the channel~\cite{horodecki1999general,gilchrist2005distance,nielsen2002simple,Watrous2018Theory}.
What it gives up is the worst-case reading: a fidelity guarantee converts to a diamond-norm one only through $\|\Phi - \hat\Phi\|_\diamond \leq d_\sfa \|\cphi - \fc(\hat\Phi)\|_1 \leq 2 d_\sfa \sqrt{2(1 - \rf(\cphi, \fc(\hat\Phi)))}$~\cite{Watrous2018Theory}, so our guarantees are complementary to the diamond-distance protocols of~\textcite{Mele2025Optimal}.

We now state the two QST algorithms we lift and the rates they yield: for concreteness, we take a non-adaptive single-copy QST algorithm~\cite{guta2020fast} and an adaptive one~\cite{chen2023does}.\footnote{Our approach lifts multi-copy QST algorithms in the same way, as long as accuracy is measured in fidelity.}
The first algorithm is reviewed in full in~\cref{Sec:QSTAlgorithms}, for the second one, we refer to~\textcite{chen2023does}. In brief, they operate as follows: 
\begin{description}[style=nextline]
    \item[$\mathcal{A}_\mathrm{QST}^\mathrm{NA}$ (Projected least squares method~\cite{guta2020fast}).] First we measure every copy with the local Pauli measurement (all settings are fixed in advance), then we form the closed-form least-squares estimate $\hrl$ from the outcome frequencies. Finally, we convert it to a density matrix by the thresholded eigenvalue truncation of~\cite{surawy2022projected} (\cref{Alg:SparseDensityEstimate}). 
    For an unknown state $\rho$ of rank at most $r$, it takes $N = O\bigl(r^2 d_\sfa^{1.6} \log(d_\sfa/\delta)/\eps^2\bigr)$ copies of $\rho$, and returns $\hat{\rho}$ such that with probability at least $1 - \delta$, $\rf (\rho, \hat{\rho}) \geq 1 - \eps$ (see \cite[Theorem~1]{guta2020fast}, \cref{Eq:GutaQST}). 
    \item[$\mathcal{A}_\mathrm{QST}^\mathrm{A}$ (Successive eigenspace projection method~\cite{chen2023does}).] 
    Over $t = O(\log(r/\eps))$ rounds, we first measure a fresh batch of copies of the unknown state $\rho$ {(of rank at most $r$) using} the uniform POVM projected onto the current residual subspace, then we learn the eigenspace of eigenvalues above the round's threshold, project it out, and recurse; finally we assemble the learned blocks and normalize.
    Given a rank bound $r$ of the unknown state $\rho$, it takes $N = O\bigl(r^2 (d_\sfa + \log(1/\delta)) \log^2(r/\eps)/\eps\bigr)$ copies of $\rho$ and returns a state $\hat{\rho}$ such that with probability $1 - \delta$, $\rf (\rho, \hat{\rho}) \geq 1 - \eps$ holds~\cite[Theorem~6.1]{chen2023does}.\footnote{Their theorem is stated for the squared fidelity, which implies the same bound for $\rf$ since $\rf \geq \rf^2$ on states.}
\end{description}

The lifted QPT protocols appear in algorithm form as~\cref{Alg:FPLSNonAdaptive,alg:FPLSadaptive}, whose QST stages are steps 1--3 and stage 1, respectively.
For simplicity we state the non-adaptive results in the multi-qubit setting, so that $\sfa$ and $\sfb$ correspond to $n_\sfa$- and $n_\sfb$-qubit spaces, $n := n_\sfa + n_\sfb = \log_2 \dimab$, and the local Pauli measurement $\mathcal P_n$ of~\cref{Alg:FPLSNonAdaptive} is defined; the lifting theorem and the analysis behind~\cref{Cor:NonAdaptSampleComp} hold for arbitrary finite-dimensional $\sfa$ and $\sfb$, with the ensemble factor $g(\dimab) = \dimab^{\log_2 3}$ of~\cref{Eq:LSConcentration} replaced by that of whichever measurement ensemble is used.

\subsection{Non-adaptive single-copy QPT} \label{Sec:NonAdaptQPT}
We begin with the \textit{non-adaptive single-copy} protocol, \textit{Fidelity Projected Least Squares} (FPLS) QPT, stated as~\cref{Alg:FPLSNonAdaptive}: it is the lift of PLS-QST~\cite{guta2020fast, surawy2022projected} applied to the Choi state of the unknown channel, with one change inside the QST stage and one at the end (summarized in~\cref{Tab:PLSvsFPLS}).
The measurement and least-squares steps of this protocol are those of PLS-QST and remain unchanged, so all data may be collected before any post-processing.
The algorithm is stated as consuming copies of the Choi state $\cphi$: each of which is obtained by preparing the maximally entangled state $\ket{\omega}$ on $\sfa \ot \sfa'$ and applying $\Phi$ on the $\sfa'$.
The least-squares estimate is Hermitian and, for certain POVMs, such as the local Pauli POVM used here, of unit trace (\cref{App:LSEstimate}). 
However, in general, it need not be PSD, so it is not a state and~\cref{Thm:Main} does not apply to it. 
Thus the protocol must first map $\hrl$ to a density matrix, and a threshold is what that step is built on.

\begin{table}[t]
\centering
\small
{\begin{tabular}{p{0.1\textwidth} p{0.15\textwidth} p{0.3\textwidth} p{0.3\textwidth}}
\hline
 & (i) LS estimate & (ii) Density estimate (\cref{Alg:SparseDensityEstimate} at threshold $\tau$) & (iii) Channel estimate via TP regularization \\
\hline
PLS-QPT \cite{surawy2022projected} & $\hrl$ (\cref{Eq:LSEstimate}) & $\tau = -\lmin(\hrl)$, below the noise level; the rank is left to the noise & HIP: iterative Frobenius quasi-projection, not rank preserving \\
FPLS-QPT (this work) & identical & $\tau = \beta_N$ (\cref{Eq:BernsteinThreshold}), above the noise level; exact rank beyond an onset & Fidelity projection $\na$ (\cref{Thm:Main}): closed form, rank preserving \\
\hline
\end{tabular}}
\caption{The three stages of PLS-QPT and FPLS-QPT. Both form the same least-squares estimate from the same local Pauli-measurement data and differ in the threshold of the density estimate and in the TP regularization.
In our numerical experiments the density estimate of PLS-QPT is computed by its authors' code, whose routine departs from~\cref{Alg:SparseDensityEstimate} in how it weights the surviving eigenvalues (\cref{App:HeadlineLadder}).}
\label{Tab:PLSvsFPLS}
\end{table}

Informally, one diagonalizes $\hrl$, discards the eigenvalues at or below the threshold $\tau \geq 0$, shifts the surviving ones up by $\tau$, and restores unit trace by adjusting the surviving (non-zero) eigenvalues (\cref{Alg:SparseDensityEstimate}); the discarded directions are where the rank of the density estimate is decided, which is why the choice of $\tau$ matters here.
We refer to~\textcite{guta2020fast} and~\textcite{surawy2022projected} for PLS-QST and PLS-QPT in full, and to~\cref{App:ProjectionStep} for the properties of this step that our proofs use.

In stage (ii) of~\cref{Tab:PLSvsFPLS}, we replace the sub-noise thresholds of PLS-QST, $\tau = 0$ or $\tau = -\lmin(\hrl)$, by the \emph{Bernstein threshold} $\tau = \beta_N$, the operator-norm concentration radius of the least-squares estimate.
The reason is rank: a \textit{supra-noise} threshold recovers the Choi rank exactly once $N$ exceeds an onset defined by the smallest nonzero Choi eigenvalue (see \cref{App:SupraNoise,App:BernsteinThreshold}), and an estimate of the right rank is closer to the true {state} in fidelity by a quadratic factor rather than a linear factor of its operator-norm error (\cref{App:Rates}), which is where the $1/\eps$ rate of~\cref{Cor:NonAdaptSampleComp} comes from.
PLS-QPT had no need of such a rule: its guarantees are stated in (trace-distance) accuracy alone, for which the sub-noise thresholds ($\tau \leq \| \rho - \hrl\|_\infty$) already inherit the least-squares rate, and its TP regularization by HIP does not preserve the rank of the density estimate in any case, so nothing would have been gained by recovering it at the QST stage.
In stage (iii), the closed form of~\cref{Thm:Main} replaces HIP.

\begin{algorithm}
    \caption{Fidelity Projected Least Squares (FPLS) QPT --- non-adaptive, single-copy}
    \label{Alg:FPLSNonAdaptive}
    
    \SetKwInOut{Input}{Input}
    \SetKwInOut{Output}{Output}
    \DontPrintSemicolon
    \Input{$N$ copies of the Choi state $\rho \equiv \cphi$ of the unknown channel $\Phi$, and confidence $\delta$; the Bernstein-threshold ($\tau = \beta_N$, \cref{App:BernsteinThreshold}) density estimator (\cref{Alg:SparseDensityEstimate}). The threshold is set from $N$, $\dimab$ and $\delta$ alone, before the data are seen.}
    \Output{A channel estimate $\hat \Phi \in \cpt$ whose accuracy $\eps$ at confidence $1 - \delta$ is given by~\cref{Cor:NonAdaptSampleComp}.}
    \tcp{1. Measure (non-adaptive; all settings fixed in advance)}
    \For{each of the $N$ copies of $\rho$}{
        Measure the copy with the local Pauli POVM $\mathcal P_n$ on $\sfa \ot \sfb$, $n = n_\sfa + n_\sfb = \log_2 \dimab$ \;
        }
        Collect the normalized outcome frequencies $c = (N_{(\vo, \vs)}/N)$. \;
        
        \tcp{2. Least-squares estimate}
        $\hrl \gets \sum_{(\vo, \vs)} c_{(\vo, \vs)} \bigotimes_{i \in [n]} \left(3 \kb{o_i, s_i} - I_2\right)$ \tcc*{\cref{Eq:LSEstimate}}
        
        \tcp{3. Density estimate: {CP with unit trace (CP1)}}
        $\hr \gets \texttt{projCP}(\hrl)$ \tcc*{\cref{Alg:SparseDensityEstimate} at $\tau = \beta_N$}
        
        \tcp{4. TP regularization: the fidelity projection}
        $\hr_\sfa \gets \tr_\sfb[\hr]$ \tcc*{if singular, see~\cref{Rem:Singular}}
        $\hrf \gets \frac1{d_\sfa}\left(\hr_\sfa \ihf \ot I_\sfb\right) \hr \left(\hr_\sfa \ihf \ot I_\sfb \right)$ \tcc*{computed cheaply via~\cref{Thm:FactorizedImplementation}}
        
        \Return $\hat \Phi = \fc \iv [\hrf]$ \;
    \end{algorithm}

    \begin{tbox}
        \begin{restatable}[FPLS-QPT: slow and fast regimes]{theorem}{NonAdaptQPTThm}\label{Cor:NonAdaptSampleComp}
            Let $\Phi \in \cpt$ have Choi rank $r$ and smallest nonzero Choi eigenvalue $\lambda_r$, and let $\delta \in (0,1)$ and $\eps \in (0, 2/d_\sfa)$.
            \Cref{Alg:FPLSNonAdaptive} at the Bernstein threshold $\tau = \beta_N$ (\cref{Eq:BernsteinThreshold}) returns $\hat\Phi \in \cpt$ with $\rf(\cphi, \fc(\hat\Phi)) \geq 1 - \eps$ with probability at least $1 - \delta$ from
            \begin{equation} \label{Eq:SlowQPTSamples}
                N = O\bigl(r^2 \dimab^{1.6} \log(\dimab/\delta)/\eps^2\bigr)
            \end{equation}
           copies of $\fc(\Phi)$.
            Moreover, if $\eps \leq {r}\lambda_r/3$, then $\rank(\fc(\hat\Phi)) = r$ and the same guarantee holds from
            \begin{equation} \label{Eq:BernsteinQPTSamples}
                N = O\bigl(r \lambda_r^{-1} \dimab^{1.6} \log(\dimab/\delta)/\eps\bigr)
            \end{equation}
            copies of $\fc(\Phi)$.
        \end{restatable}
    \end{tbox}
    Explicitly, $N \geq \tfrac{512}{3} r^2 \dimab^{1.6}\log(\dimab/\delta)/\eps^2$ suffices for~\cref{Eq:SlowQPTSamples} and $N \geq \tfrac{512}{9} r \lambda_r^{-1} \dimab^{1.6} \log(\dimab/\delta)/\eps$ for~\cref{Eq:BernsteinQPTSamples}; \cref{Sec:AccuracyRate} shows the constants are loose by orders of magnitude.\footnote{The exponent of $\dimab$ is set by the measurement ensemble: for the local Pauli measurements used throughout it is $\log_2 3 \approx 1.585 \leq 1.6$, which we round up as~\textcite{surawy2022projected} do, and for the uniform POVM the factor $\dimab^{1.6}$ is replaced by $2 \dimab$~\cite{guta2020fast}.}
    In both bounds, $r$ and $\lambda_r$ are parameters of the bound{;} they are not inputs to the algorithm.
    {The ingredients of the proof are developed in~\cref{Sec:QSTAlgorithms}; the proof itself is in~\cref{App:QPTProof}.}
    \paragraph{The role of each component.}
    The fast rate needs both departures from PLS-QPT listed in~\cref{Tab:PLSvsFPLS}.
    \emph{Threshold:} at $\tau = 0$ or $\tau = -\lmin$, the rank of $\hr$ is left to the noise (\cref{Fig:ThresholdStrategies}) and the operator-norm error converts to infidelity only linearly (\cref{Eq:RankConversion}), so the rate is the slow one (\cref{Eq:SlowQPTSamples}) whichever regularizer follows.
    The Bernstein threshold instead recovers the rank beyond a finite onset (\cref{Cor:BernsteinRank}), which upgrades the conversion to the quadratic one of~\cref{Prop:QuadraticConversion}.

    \emph{Regularizer:} the fidelity projection carries that quadratic bound from the density estimate to the channel estimate (\cref{Thm:Main}) and preserves its rank.
    Preserving the rank is sufficient for the fast rate but not necessary: it suffices that the weight the estimate places outside the (estimated approximation of the) support of $\cphi$ be quadratically small in its operator-norm error (\cref{Rem:OffSupportWeight}).
    The thresholds $\tau \in \{0, - \lmin\}$ do not meet even that weaker condition.
    Moreover, HIP returns high Choi rank in~\cref{Sec:Numerics}, and on those instances the weight it leaves outside the support is first order in $\Delta$ rather than second.

    \emph{Onset:} before it, the threshold clips true eigenvalues of $\cphi$, the rank falls short of $r$, and only the slow bound applies (\cref{Lem:SupraNoiseStability}); the clipped weight shows as a plateau in~\cref{Fig:ThresholdStrategies,Fig:Headline}.
    The onset is the price of the fast rate, and $\lambda_r$ sets it: for well-conditioned Choi spectra, where $\lambda_r = \Theta(1/r)$, the fast bound~\cref{Eq:BernsteinQPTSamples} reads $N = O(r^2 \dimab^{1.6} \log(\dimab/\delta)/\eps)$, which matches the $\eps$-dependence of the adaptive protocol of~\cref{Cor:AdaptSampleComp} without adaptivity.
    
    \subsection{Adaptive single-copy QPT} \label{Sec:AdaptQPT}
    For the adaptive setting, in the QST stage, we consider the adaptive algorithm of~\cite{chen2023does}, and the lifted QPT protocol carries the following guarantee.
    
    \begin{tbox}
        \begin{cor}[Adaptive QPT] \label{Cor:AdaptSampleComp}
            Let $\Phi \in \cpt$ be the unknown channel with Choi state $\rho \equiv \fc(\Phi) \in \cab$, let $\gamma, \delta \in (0,1)$, and let $r \geq \rank(\rho)$ be a known rank bound. For $\gamma < 2/d_\sfa$, \cref{alg:FPLSadaptive}, given
            \begin{equation}
                N = O\Lt(\frac{r^2}{\gamma} \left(\dimab + \log(1/\delta)\right) \log^2(r/\gamma)\Rt)
            \end{equation}
            copies of $\fc(\Phi)$, produces an estimate $\tilde{\rho} \in \cab$ that satisfies $\rf(\rho, \tilde{\rho}) \geq 1 - \gamma$ with probability at least $1-\delta$.
        \end{cor}
    \end{tbox}
    \begin{proof}
        We apply~\cref{Thm:Main}(3) to the rank-$r$ tomography algorithm of~\textcite[Theorem~6.1]{chen2023does}, run at accuracy $\gamma/4$. 
    \end{proof}
    
    Explicitly, $N \geq 4c\, r^2 (\dimab + \log(1/\delta)) \log^2(4r/\gamma)/\gamma$ suffices, where $c$ is the universal constant from \textcite[Theorem~6.1]{chen2023does}.
    Since the lifting only changes the accuracy parameter, any improvement of $c$, or an improved QST algorithm, is inherited by the QPT algorithm with no change to our argument.
    The measurement model here differs from the non-adaptive setting in an important way. 
    Each round measures with a \textit{projected uniform POVM} on the current residual subspace. Consequently, it cannot be handed a data set collected in advance --- round $j$'s measurement depends on the outcomes of rounds $1,\dots,j-1$ --- so the input is a supply of \textit{fresh} Choi-state copies consumed batch by batch. These requirements make the protocol harder to realize than~\cref{Alg:FPLSNonAdaptive}.
    \begin{algorithm}
        \caption{Adaptive single-copy QPT}
        \label{alg:FPLSadaptive}
        
        \SetKwInOut{Input}{Input}
        \SetKwInOut{Output}{Output}
        \DontPrintSemicolon
        \Input{Copies of the Choi state $\rho \equiv \cphi$ on demand, accuracy $\gamma$, confidence $\delta$, and a rank bound $r$.}
        \Output{A channel estimate $\hat \Phi \in \cpt$ with Choi state $\tilde \rho$ satisfying $1 - \rf(\rho, \tilde{\rho}) \leq \gamma$ with probability at least $1-\delta$.}
        
        Set the number of rounds $T$ and the per-round batch size $m$ as in~\cite{chen2023does}, for target state-accuracy $\gamma/4$ {(the factor $4$ is the lifting loss of~\cref{Thm:Main})} \;
        Initialize the residual projector $\Gamma_1 \gets I_{\sfa \ot \sfb}$ \;
        
        \tcp{1. Adaptive QST of~\cite{chen2023does}: measurements depend on earlier outcomes}
        \For{$j = 1, \dots, T$}{
            Draw a \textit{fresh} batch of $m$ copies of $\rho$ \;
            Measure each with the projected uniform POVM on $\Gamma_j$, and form the projected estimator $\sigma_j$ \;
            $\Pi_j \gets$ spectral projector of $\sigma_j$ onto eigenvalues $\geq 2^{-j}$ \tcc*{resolved eigenspace}
            $\Gamma_{j+1} \gets I - \sum_{i \leq j} \Pi_i$ \tcc*{restrict to the orthogonal complement}
            }
            Assemble the block estimate $\hat S \gets \sum_{j=1}^T \Pi_j \sigma_j \Pi_j$ and normalize: $\hr \gets \hat S / \tr[\hat S]$ \;
            
            \tcp{2. TP regularization: the fidelity projection}
            $\hr_\sfa \gets \tr_\sfb [\hr]$ \tcc*{if singular, see~\cref{Rem:Singular}}
            $\tilde \rho \gets \frac1{d_\sfa}\left(\hr_\sfa \ihf \ot I_\sfb\right) \hr \left(\hr_\sfa \ihf \ot I_\sfb\right)$ \tcc*{\cref{Eq:ClosedFormProj}}
            
            \Return $\hat \Phi = \fc \iv [\tilde \rho]$ \;
        \end{algorithm} 
\section{Numerical experiments} \label{Sec:Numerics}
We now present numerical results that complement the theory; code and data reproducing every figure and table are available in the accompanying repository~\cite{afham2026fplscode}.
We study only FPLS, the non-adaptive algorithm of~\cref{Sec:NonAdaptQPT}, since it needs only single-copy local Pauli measurements with settings fixed in advance, and compare it with PLS-QPT~\cite{surawy2022projected} (run with its authors' code, at their default settings wherever its cost is reported; see~\cref{App:HeadlineLadder}), the only other non-adaptive single-copy QPT algorithm that is constructed via lifting and carries rigorous guarantees.
Both follow the three stages of~\cref{Tab:PLSvsFPLS} on the same measurement data.
A reported rank is a \emph{numerical} rank, the number of eigenvalues with absolute value above $\eps_{\mathrm{zero}} = 10^{-10}$, and the feasibility tolerances quoted for HIP and for the SDP are floating-point residuals rather than exact membership of $\cab$.

Throughout, FPLS thresholds at $\tau = \beta_N/2$, half the Bernstein radius at confidence parameter $\delta = 0.05$, with the curves at $\tau = \beta_N$ of~\cref{Cor:NonAdaptSampleComp} shown wherever the two differ.
The halved threshold is a heuristic that brings the onset forward; the guarantee of~\cref{Cor:NonAdaptSampleComp} holds only for $\tau = \beta_N$.
On every instance reported here it nevertheless satisfies the hypothesis $\tau \geq \|\hrl - \rho\|_\infty$ of~\cref{Lem:SupraNoiseStability}, the realized noise never exceeding $0.29 \beta_N$, so the estimate never over-counts the rank and recovers it exactly whenever $\beta_N < \lambda_r$, a quarter of the samples required by~\cref{Cor:BernsteinRank}.
Verifying that hypothesis requires the true state, so an experimenter cannot certify it for their data.

The diamond-distance projection of~\textcite{Mele2025Optimal} appears only in~\cref{Fig:Headline}~(left): it is an SDP whose cost grows as roughly $\dimab^3$, which puts it out of reach at the dimensions and shot counts studied here.
We solve it with \textsc{Clarabel} and \textsc{SCS} and report the solver's own solve time (\cref{App:HeadlineLadder}), which excludes the modelling overhead and so understates its cost; HIP and the fidelity projection are timed end to end.\footnote{The authors' implementation of HIP logs intermediate results at every iteration; this logging is not included in the reported time.}
HIP, an iterative scheme whose unit of work is one eigendecomposition of a $\dimab \times \dimab$ matrix, can be compared in iterations as well as in seconds, and we report both in~\cref{Sec:CostTP}.
Its cost is reported at its authors' default stopping tolerance; a tighter tolerance buys accuracy with more iterations (\cref{App:HeadlineLadder}), a trade-off that the closed-form fidelity projection, being exact, does not need to pay.

When the underlying channel is low rank, FPLS-QPT improves on every metric of interest.
Time and space concern stage~(iii) alone; accuracy reflects stages~(ii) and~(iii).
\begin{itemize}
    \item \emph{Time} (\cref{Sec:CostTP}), ahead throughout: one $d_\sfa \times d_\sfa$ eigendecomposition replaces a number of eigendecompositions of size $\dimab$ that grew as $\dimab^{0.4}$ over the tested range, a gap of four orders of magnitude at $\dimab = 2^{10}$ (\cref{Tab:HeadToHead}).
    \item \emph{Space} (\cref{Sec:CostTP}), ahead whenever the channel is low rank: the estimate is held as $r$ Kraus operators, $O(r \dimab)$ numbers, against $O(\dimab^2)$ for HIP.
    \item \emph{Accuracy} (\cref{Sec:AccuracyRate,Sec:ChannelRobustness}): past the empirical onset of rank recovery our infidelity converges at $N^{-1}$ against $N^{-1/2}$ for PLS-QPT, an advantage growing to more than two orders of magnitude, while in trace and Frobenius distance both converge at $N^{-1/2}$ and we lead by a constant below $2$; before the onset we are behind, by up to $10\times$ on the instance of~\cref{Fig:Headline}.
\end{itemize}

\emph{The Bernstein threshold captures the true Choi rank of the channel, and the rank preservation of the fidelity projection lets the channel estimate inherit it}; both are needed for the fast rate, for the reasons given after~\cref{Cor:NonAdaptSampleComp}.
HIP, a \textit{quasi-projection}\footnote{We term it a quasi-projection as it is not a true Frobenius projection, but is contractive, like a true Frobenius projection.}~\cite[Sec.~4]{surawy2022projected}, does not preserve the rank: its Frobenius projection onto the TP hyperplane~\cite[Section 13.1]{surawy2022projected} and its final mixing step~\cite[Eq.~(15)]{surawy2022projected} yield estimates with high or nearly full Choi rank throughout our experiments.
A high numerical rank does not by itself imply the slow rate, which is decided by the off-support mass $\ell$ of~\cref{Eq:OffSupportMass}; on our instances, however, it comes with a slower infidelity slope.
For channels of (nearly) full Choi rank the preference for low-rank estimates backfires and PLS-QPT dominates over a wide range of $N$ (\cref{Sec:ChannelRobustness}, \cref{Fig:FullRankTradeoff}).

\subsection{Partial normalization vs HIP: Time and space costs} \label{Sec:CostTP}
We now compare the time and space costs of TP regularization by fidelity projection and by HIP.
If the density estimate $\hr$ is given in eigendecomposed form (equivalently, as a Kraus representation of the CP map $\fc\iv(\hr)$), a Kraus representation of the channel estimate $\pc(\hr)$ is {inexpensive to compute}.
\begin{tbox}
\begin{restatable}[Fidelity projection in Kraus-form]{theorem}{fidelityprojectioninKrausform} \label{Thm:FactorizedImplementation}
        Let $\Theta \in \cpab$ have a Kraus representation $\{K_i\}_{i=1}^r$, so that $\Theta(\si) = \sum_{i\in[r]} K_i \si K_i \dg$, and let $\hr = \fc(\Theta)$ be its Choi state. Let $R := \sum_{i \in [r]} K_i \dg K_i \in \posa$
        and assume $R > 0$.
        Then $\pc[\hr]$ is the Choi state of the channel with Kraus operators
        $\{K_i' = K_i R \ihf\}_{i \in [r]}$
        which can be computed in
        \begin{equation}
            O\!\left(r d_\sfa^2 d_\sfb + d_\sfa^3\right) \quad \text{arithmetic operations and} \quad O\!\left(r \dimab + d_\sfa^2\right) \ \text{memory},
        \end{equation}
        without ever forming a $\dimab \times \dimab$ matrix.
\end{restatable}
\end{tbox}
\noindent The proof is in~\cref{App:FactorizedProof}.

Two properties of the fidelity projection enable this advantage: a closed-form solution and a factorized implementation of it.
HIP, in contrast, performs an eigendecomposition of a {$\dab \times \dab$} Hermitian matrix at each iteration, at cost $O(\dab^3)$, and the number of iterations it needs is not known a priori. 
In \cref{Fig:HIPPerformanceAnalysis}, we study the iteration complexity of HIP against dimension.
Every (unprojected) instance there is calibrated to approximately a fixed fidelity from $\cab$,\footnote{Fidelity of a state $\rho$ to the set $\cab$ is defined as $\max_{\sigma \in \cab}\rf(\rho, \sigma)$.} and this calibration protocol and the true channels are described in~\cref{App:SyntheticInstances}.

\begin{figure}[htbp]
     \centering
     \includegraphics[width=\textwidth]{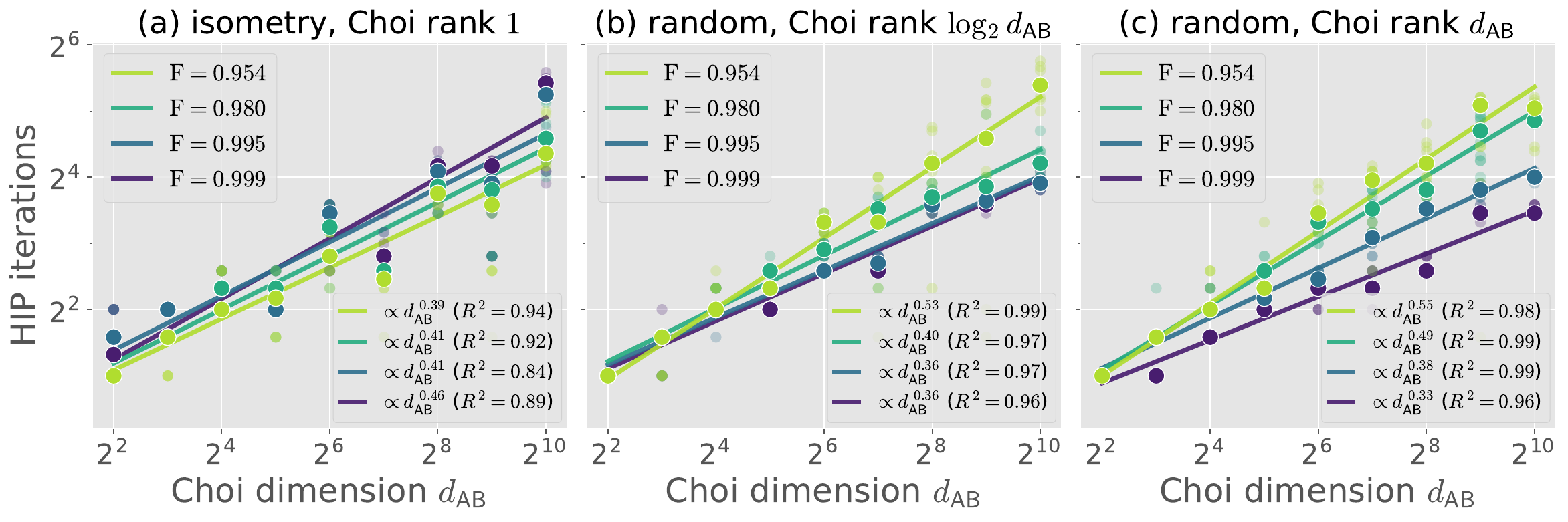}
     \caption{Iterations required by HIP~\cite{surawy2022projected}, run as described in~\cref{App:HeadlineLadder}, against Choi dimension, every instance calibrated approximately to a common fidelity $\mathrm{F}$ from $\cab$, at four levels colored from dark (closest to $\cab$) to light (farthest) and labeled by the corresponding fidelity. 
     True channels are (a) a Haar-random isometry, Choi rank $1$, and (b), (c) random channels of Choi rank $\log_2 \dimab$ and $\dimab$, a Wishart matrix normalized by its input marginal~\cite{bruzda2009random}; ten measurement draws each. 
     Translucent markers are single instances, opaque markers their median. Solid lines are log--log power-law fits to the medians.}
     \label{Fig:HIPPerformanceAnalysis}
\end{figure}

For isometric channels, over $\dimab \in [4, 1024]$ and four fidelity levels $\rf =$ 0.999, 0.995, 0.980, 0.954, the power-law fits to the medians are
\begin{equation} \label{Eq:HIPIterScaling}
    \#\text{iterations} \approx 1.3\, \dimab^{\,0.46}, \quad 1.5\, \dimab^{\,0.41}, \quad 1.3\, \dimab^{\,0.41}, \quad 1.2\, \dimab^{\,0.39}.
\end{equation}
It is therefore not only the cost of an iteration that grows with dimension, which is expected, but the \emph{number} of them: over the tested range it grows as $\dimab^{0.4}$.
Indeed the Choi rank of the true channel also plays a role, and where it does is the point.
The fitted exponent $b$ of $\#\text{iterations} \propto \dimab^{\,b}$, for the three channel families of~\cref{Fig:HIPPerformanceAnalysis} at each of the four calibration levels, is summarized below.
\begin{center}
    \small
    \begin{tabular}{lcccc}
        \hline
        & \multicolumn{4}{c}{fidelity of the estimate to $\cab$} \\
        Choi rank of the channel & $0.999$ & $0.995$ & $0.980$ & $0.954$ \\
        \hline
        $1$ (isometry) & $0.46$ & $0.41$ & $0.41$ & $0.39$ \\
        $\log_2 \dimab$ & $0.36$ & $0.36$ & $0.40$ & $0.53$ \\
        $\dimab$ (full rank) & $0.33$ & $0.38$ & $0.49$ & $0.55$ \\
        \hline
    \end{tabular}
\end{center}
\begin{itemize}
    \item For estimates close to $\cab$, the higher the Choi rank the slower the iteration count grows with dimension: the first column falls from $0.46$ to $0.33$.
    \item Far from $\cab$ it is the isometry that grows slowest: the last column reads $0.39$ against $0.53$ and $0.55$ for the two random families.
    This is caused by HIP's stopping rule, as its default tolerance is proportional to the noise of the least-squares estimate, $\epsilon_{\mathrm{tol}} = \max\bigl(10^{-1} |\lmin(\hrl)|, 10^{-3}\bigr)$ (\cref{App:HeadlineLadder}), so a noisier estimate is regularized less precisely.
    At a fixed distance from $\cab$, the estimate of an isometry comes from far fewer shots than that of a higher-rank channel, $3.5 \times 10^4$ against $9.1 \times 10^5$ for the full-rank channel at $\dimab = 2^{10}$, and is correspondingly noisier: $\epsilon_{\mathrm{tol}} \approx 0.1$ against $0.02$.
    For the isometry this tolerance does not tighten as the dimension grows; for the other two families it does, which adds iterations at the large dimensions and steepens their fits.
\end{itemize}
These numerics further validate our method, especially when the true channel is of low Kraus-rank and the estimate is close to $\cab$, as  HIP requires more iterations for convergence in this setting. 
Observe that the entire cost of the fidelity projection is thus below that of a single HIP iteration: its only eigendecomposition is of the $d_\sfa \times d_\sfa$ marginal defect $R := \sum_{{i \in [r]}} K_i^\dagger K_i$, so for $d_\sfa = d_\sfb$ and $r = O(1)$, it costs $O(\dimab^{3/2})$ against HIP's $O(\dimab^3)$ \emph{per iteration}, in $O(\dimab)$ memory rather than $O(\dimab^2)$.

This compounding advantage produces the gap in~\cref{Fig:Headline}~(left), which reports the time cost on the same calibrated instances, medians over ten instances.
At $\dimab = 2^{10}$, regularization takes HIP $11.5$ seconds against $0.28$ milliseconds for the rank-factorized projection; the diamond-norm projection of~\textcite{Mele2025Optimal}, an SDP growing as roughly $\dimab^{3}$, takes $19$ seconds at $\dimab=2^6$ and $86$ seconds at $\dimab=2^7$ and extrapolates to around ten hours at $\dimab=2^{10}$.
{We expect the gap in~\cref{Fig:Headline}~(left) to widen beyond $\dab = 2^{10}$: HIP costs one $\dimab$-sized eigendecomposition per iteration, $O(\dimab^3)$, times a number of iterations that grew as $\dimab^{0.4}$ over the tested range (\cref{Eq:HIPIterScaling}), about $\dimab^{3.4}$ in total if that trend continues.} 
The factorized projection of~\cref{Thm:FactorizedImplementation} is not on this scale at all: its cost grows linearly with the Choi rank $r$, so the advantage over HIP is largest for low-rank channels and shrinks as $r$ grows, but even at full rank it stays below a single HIP iteration by a factor $d_\sfb$.
The space advantage follows from the same theorem.
FPLS holds its estimate as $r$ Kraus operators of size $d_\sfb \times d_\sfa$ and touches only the $d_\sfa \times d_\sfa$ defect $R$, $O(r \dimab + d_\sfa^2)$ numbers in all.

Thanks to the rank-preservation of {the} fidelity projection, the channel- and the density-estimates have the same rank (\cref{Fig:SpectrumMatching}), so a low-rank channel is returned, stored and applied using a low-rank representation.
In contrast, HIP must hold and eigendecompose the full $\dimab \times \dimab$ matrix at every iteration, handling $O(\dimab^2)$ numbers, and its output is of high rank, so it cannot be compressed afterwards without a further approximation.
At $\dimab = 2^{10}$ and $r = 2$, the returned estimates hold $2 \times 10^3$ against $10^6$ complex numbers, $33$ kB against $16.8$ MB in double precision, and the ratio grows as $\dimab / r$.

\subsection{Accuracy and the observed convergence rate} \label{Sec:AccuracyRate}
\begin{figure}[!t]
    \centering
    \includegraphics[width=\textwidth]{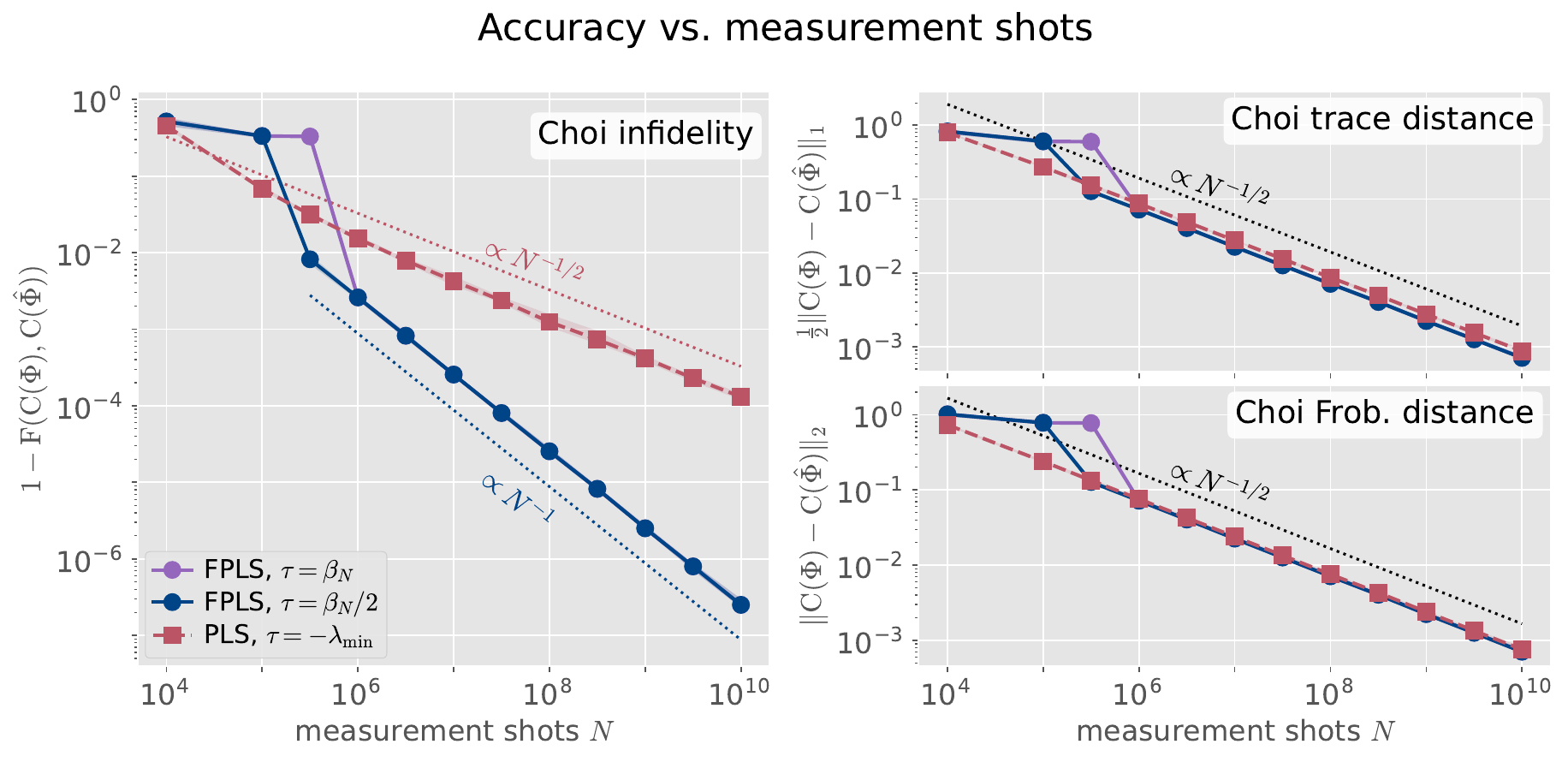}
    \caption{Reconstruction accuracy against the shot count for a rank-$2$ channel of~\cref{Fig:Headline} at $\dimab = 2^8$, one panel per metric. 
    FPLS at the Bernstein threshold $\tau = \beta_N$ (purple) and at $\beta_N/2$ (blue) against PLS at $\tau = -\lmin$ (red); {geometric means} and $[\min,\max]$ bands over $10$ trials. {The dotted guides carry $N^{-1/2}$ and $N^{-1}$ with fixed, never fitted exponents, drawn as in~\cref{Fig:Headline} (\cref{App:HeadlineRight})}.}
    \label{Fig:AccuracyVsN}
\end{figure}

\Cref{Fig:AccuracyVsN} presents the predicted rates in three metrics on the (rank-2 mixed-unitary) channel studied in~\cref{Fig:Headline} at $\dimab = 2^8$.
In infidelity, the thresholded estimates show the rank-recovery onset that~\cref{Cor:BernsteinRank} guarantees from $N = 3 \times 10^6$ at $\tau = \beta_N$; empirically it arrives earlier, at $N = 10^6$ for $\tau = \beta_N$ and at $N = 3 \times 10^5$ for $\tau = \beta_N/2$.
Below its empirical onset, every trial returns rank $1$: a true eigenvalue is discarded and the infidelity sits on a plateau, behind PLS-QPT by $10.5\times$ and $4.9\times$ at the last shot count before the respective onsets.
From the onset on, every trial returns rank $2$ and the two thresholded estimates coincide.
The infidelity then drops by two orders of magnitude and converges at a fitted exponent $-1.00$, the $N^{-1}$ of~\cref{Eq:BernsteinQPTSamples}, while PLS-QPT continues at $-0.51$.
The ratio grows from $6\times$ at $N = 10^6$ to $49 \times$ at $N=10^8$ and $520 \times$ at $N=10^{10}$.

In trace and Frobenius distances, the picture is different.
All three estimates converge at fitted exponents of $-0.50$ in both norms, and the thresholded estimate is ahead only by a constant, $1.2\times$ in trace distance and $1.1\times$ in Frobenius distance at $N = 10^{10}$; before the empirical onset it is behind by up to $2.2\times$ and $3.3\times$.
Rank recovery thus buys a faster rate in fidelity, and the reason is a property of the fidelity rather than of the estimator (see \cref{Prop:QuadraticConversion}).
The infidelity lines of \cref{Fig:Headline,Fig:AccuracyVsN,Fig:ChannelRobustness,Fig:FullRankTradeoff} are geometric means over trials.
\Cref{Fig:ThresholdStrategies} instead draws one median per branch for $\tau = -\lmin(\hrl)$, whose trials under~\cref{Alg:SparseDensityEstimate} are bimodal, the rank being recovered in about half of them.

\subsection{Performance across channel families} \label{Sec:ChannelRobustness}
We conclude this section by comparing FPLS- and PLS-QPT on more channels (all of which are 3-qubit channels with $\dab = 2^6$), with the metrics being rank-recovery and accuracy (infidelity) as we have done previously. 
We consider two classes of channels: channels without and with full Choi rank. 

\paragraph{Channels without full Choi rank.}
\begin{figure}[!h]
    \centering
    \includegraphics[width=\textwidth]{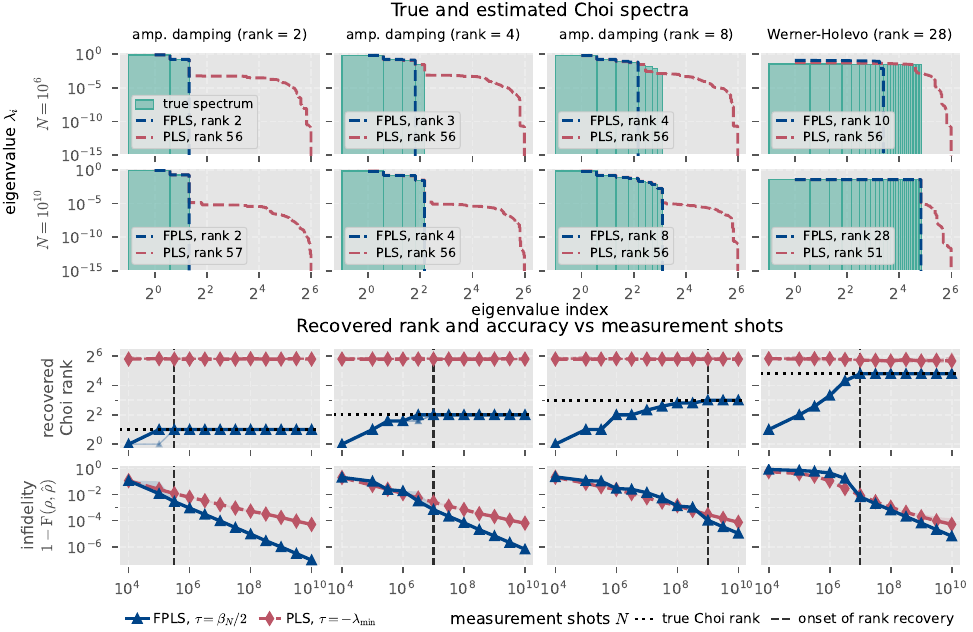}
    \caption{Rank and accuracy of the estimates across channels at $\dimab = 2^6$: independent amplitude damping on $k$ qubits (Choi rank $2^k$) and the Werner--Holevo channel (Choi rank $28$). Top two rows, the true Choi spectrum (teal bars) against the spectra recovered by FPLS (blue dashed line) and PLS (red dashed line) at $N = 10^{6}$ (first row) and $N = 10^{10}$ (second row), each panel's legend giving the ranks of the two estimates; third row, recovered Choi rank, translucent lines the ten trials, opaque line their mode, dotted line the true rank; last row, infidelity, geometric means over $[\min,\max]$ bands. Blue is FPLS at the Bernstein threshold $\tau = \beta_N/2$, red PLS; the dashed vertical line marks the empirical onset, the smallest plotted shot count from which every FPLS trial returns the true rank.}
    \label{Fig:ChannelRobustness}
\end{figure} 
We consider the tensor-product of $k$ single-qubit amplitude damping channels applied on $k \in \{1, 2, 3\}$ qubits, which leads to Choi states with rank $2^{k}$.
The damping strengths {$\phi_1, \dots, \phi_k$} are drawn in $[0.15, 0.45]$, each with Kraus operators $A_0 = \kb{0} + \sqrt{1 - {\phi_j}}\,\kb{1}$ and $A_1 = \sqrt{{\phi_j}}\,\ketbra{0}{1}$, so that the Choi eigenvalues are the products of $1 - {\phi_j}/2$ and ${\phi_j}/2$ over the $k$ qubits. 
We also consider the \emph{Werner--Holevo} channel~\cite{werner2002counterexample} $\rho \mapsto (\tr(\rho) I - \rho^{\mathsf T})/(d_\sfa - 1)$, whose Choi state is proportional to the projector onto the antisymmetric subspace and has a flat spectrum of rank $d_\sfa(d_\sfa - 1)/2 = 28$.
The Werner--Holevo channel is chosen for its non-trivial flat spectrum.
In the spectra of~\cref{Fig:ChannelRobustness} (top two rows), at $N = 10^{10}$, FPLS returns the true rank exactly in every column, whereas PLS returns $51$ to $57$ of a possible $64$ directions, most of them carrying weights below $10^{-5}$; at $N = 10^{6}$, FPLS under-counts on every channel except the rank-$2$ one, while PLS returns $56$ throughout ($48$ to $63$ over all shot counts).

The third row plots rank against measurement shots, which shows that the empirical (rank-recovery) onset moves with the smallest nonzero Choi eigenvalue $\lambda_r$, as the guaranteed onset of~\cref{Cor:BernsteinRank} does: FPLS returns the true rank in every trial from $N = 3 \times 10^5$ at rank $2$, $N=10^7$ at rank $4$, $N=10^9$ at rank $8$ and $N=10^7$ for the Werner--Holevo channel, and below the onset it always undercounts.
The onset is proportional to $1/\lambda_r^2$: the Werner--Holevo channel, of rank $28$ but with a flat spectrum, $\lambda_r = 1/28$, recovers its rank two decades before amplitude damping of rank $8$, whose eighth eigenvalue is smaller.
The accuracy panels follow the rank panels: past the onset, FPLS converges at fitted exponents $-1.0$ on all four channels against $-0.5$ to $-0.7$ for PLS, its advantage growing with $N$.
Before the onset, it is behind PLS, the price a sparse estimator pays for discarding true weight: by up to $2.7\times$ on the rank-$8$ channel and $6.0\times$ on the Werner--Holevo channel, though never on the rank-$2$ channel, where FPLS is ahead at every shot count.

\begin{figure}[h]
    \centering
    \includegraphics[width=\textwidth]{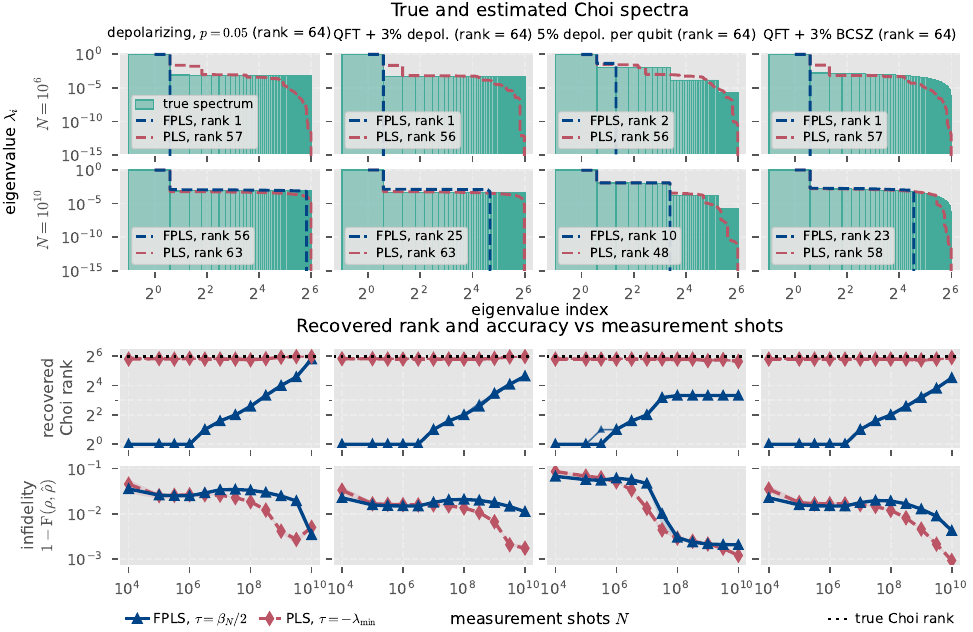}
    \caption{
    Channels with full Choi rank at $\dimab = 2^6$: depolarizing at $p = 0.05$ and the {quantum Fourier transform (QFT)} followed by global depolarizing noise at $p = 0.03$, both of full Choi rank $64$ with flat tails; then independent $5\%$ depolarizing on each of the three qubits, of full rank with a tiered tail; and the convex combination $0.97\,\Phi_{\mathrm{QFT}} + 0.03\,\Phi_{\mathrm{BCSZ}}$ of the QFT channel with a random channel drawn from the BCSZ distribution~\cite{bruzda2009random} of full rank with a spread tail.
    Rows, colors and rank cutoff are as in~\cref{Fig:ChannelRobustness}; no dashed onset line appears, since none of the four channels recovers its rank in every trial within $N \leq 10^{10}$.
    \label{Fig:FullRankTradeoff}
    }

\end{figure}
\paragraph{Channels with full Choi rank.}
We now consider channels with full Choi rank, the setting where the affinity of FPLS towards low-rank {estimates} can backfire. 
Here we consider {the} global depolarizing channel (at noise $p = 0.05$), the Quantum Fourier Transform (QFT) followed by a global depolarizing channel (at noise $p = 0.03$), the tensor-product of single-qubit depolarizing channels (each at noise $p = 0.05$), and a (0.97, 0.03)-convex combination of a QFT channel and a random channel of full Choi rank drawn from the BCSZ distribution~\cite{bruzda2009random}.
The first and third are perhaps the most well-studied channels with full Choi rank. 
The second and fourth represent a unitary channel with isotropic and anisotropic noises, respectively. 
With a gap too small to place the threshold in (at the shot counts considered here) the Bernstein rule admits a direction (eigenvector) only once the data resolves it above the threshold $\tau$, so the rank FPLS returns climbs with $N$.
PLS keeps every direction above the noise edge (at the density estimation stage) and its TP step (via HIP) spreads the remaining weight over the rest, so it returns estimates with ranks $48$ to $63$ which, in this setting of true channels with full Choi rank, is advantageous.

\section{Conclusion} \label{Sec:Conclusion}
In this article we introduced a new approach to lift any fidelity-based QST algorithm to a QPT algorithm, while allowing for the inheritance of sample complexity guarantees. 
This is enabled by the identification of a widely used \textit{partial normalization} operation as the exact fidelity projection of a bipartite density matrix onto the set of Choi states of channels~\cite{afham2026projections}.  
The partial normalization comes with a variety of benefits over existing TP regularization techniques---it is exact, closed-form, and can be computed inexpensively in a rank-factorized manner. 
It is also rank-preserving, which allows for construction of QPT algorithms that can return channel estimates with low Choi rank. 
We then {presented} two concrete (single-copy) QPT algorithms--- one non-adaptive and one adaptive---by lifting the QST algorithms of~{\cite{guta2020fast, chen2023does}}. 
Our non-adaptive QPT algorithm, called \textit{Fidelity Projected Least Squares} QPT, excels at the tomography of channels with low Choi rank: beyond a \textit{rank recovery} onset (in the number of samples), it achieves a faster rate of convergence and exactly captures the Choi rank of the channel{, while its TP-regularization step runs orders of magnitude faster and leaner than the iterative alternative of~\cite{surawy2022projected}}.

\section*{Statement on AI usage}
The authors acknowledge the use of Large Language Models (LLMs), Claude Opus 5 for most of the project and Claude Fable 5.1 and ChatGPT  for the final revision, in the preparation of this article.
This includes standard usage such as proofreading, restructuring, and basic questions.
The simulation and plotting scripts were written with the LLM; the core routines for projections, estimators and random channels were written by the authors, and the HIP baseline is the implementation of~\cite{surawy2022projected, kahn2021hip}, run as described in~\cref{App:HeadlineLadder}.

The project started towards the end of 2025, and the key idea of this article, using the closed-form fidelity projection to lift fidelity-based QST algorithms to QPT algorithms, was contributed by the authors.
Key contributions of the LLM are the formalization of the Bernstein-threshold analysis (\cref{App:ProjectionStep} to \cref{App:QPTProof}), the marginal-invertibility lemma (\cref{App:Invertibility}), and the calibrated benchmark methodology of~\cref{App:SyntheticInstances}.
The Bernstein-threshold analysis was arrived at as follows.
From the results of the numerical experiments, the authors noticed that accurate capture of the Choi rank leads to a faster decay of the infidelity, and it was noticed that the $\tau = -\lambda_{\min}$ rule recovers the rank in only about half the trials at every shot count.
The authors then reasoned that a thresholding strategy that uses the concentration radius should capture the true rank beyond some onset, and the LLM was asked to formalize this idea and provided the proofs of \cref{App:ProjectionStep} to \cref{App:QPTProof}.
The authors have verified the correctness of and take responsibility for the results in this article.

\paragraph{Acknowledgments}
AA and SS would like to thank Jan Seyfried for discussions during the project.
AA also thanks Roberto Rubboli for helpful discussions on these topics. 
AA, SS, and MT are supported by the NRF Investigatorship award (NRF-NRFI10-2024-0006).

\printbibliography

\appendix
\crefalias{section}{appendix}
\crefalias{subsection}{appendix}

\newpage
\color{black}

\section{Projected least squares at the Bernstein threshold} \label{Sec:QSTAlgorithms}
This appendix builds, step by step, the tools behind~\cref{Cor:NonAdaptSampleComp}, and proves it in the last subsection (\cref{App:QPTProof}).
Throughout this appendix, $d$ denotes the dimension of the state space of the QST problem; for the Choi states of~\cref{Sec:Lifting}, $d = \dimab$.
Informally, the proof strategy is as follows.
\begin{enumerate}
    \item First we use the fact that the (Hermitian) {least-squares} estimate $\hrl$ is a close approximation of the true state $\rho$, in spectral-norm distance (\cref{App:LSEstimate}).
    \item Then we show that the thresholded-projection (to obtain the density estimate $\hr$) remains stable under spectral-norm distance (\cref{App:ProjectionStep}).
    \item Next we prove that the thresholding rule we call the \textit{Bernstein threshold} captures the true underlying rank beyond a certain onset of sample size (\cref{App:SupraNoise,App:BernsteinThreshold}). 
    \item Thereafter, we show that accurately capturing the rank leads to an improved rate in infidelity decay (\cref{App:Rates}). 
    \item Finally, we combine all of the above to obtain the sample complexity (\cref{App:QPTProof}). 
\end{enumerate}

The QST algorithm we lift is a non-adaptive, single-copy QST algorithm called the \textit{Projected Least Squares} QST algorithm~\cite{guta2020fast,surawy2022projected}.
It has two components: a \emph{least-squares} stage, which turns the measurement outcomes into a Hermitian estimate of the state (\cref{App:LSEstimate}), and a \emph{projection} stage, which turns that estimate into a density matrix (\cref{App:ProjectionStep}); the rank and accuracy guarantees of the paper are properties of how the second stage is carried out.
The copies of the true state $\rho \in \dna$ are measured with a \textit{structured POVM}, that is, a POVM whose frame operator can be inverted in closed form, so that the LS estimate has a closed-form expression; examples are the uniform POVM and local Pauli measurements.
When their elements have equal trace, as both of these do, the LS estimate is moreover guaranteed to have unit trace. 
For simplicity, we focus on local Pauli measurements.

\subsection{The least-squares (LS) estimate} \label{App:LSEstimate}
\paragraph{The LS estimate for a general POVM.}
For the sake of completeness, we now discuss the closed form of the LS estimate, first for a general POVM and then for local Pauli measurements.
This discussion closely follows the one from~\cite{guta2020fast}.
Let $\rho \in \dna$ be the true state and let $\cle = \{E_i\}_\inm$ be a POVM associated with $\sfa$. 
The associated measurement channel (denoted as $\mathrm{M}_\cle$) and its adjoint are defined as 
\begin{equation}
    \mathrm{M}_\cle(X) = \sum_\inm  \la E_i, X \ra |i \ra \qaq \me^\dagger(v) = \sum_\inm v_i E_i, 
\end{equation}
for any $X \in \mathrm{L}(\sfa)$ and $v \in \mathbb{C}^m$. 
In (single-copy, non-adaptive) QST, one is given $N$ i.i.d. copies of $\rho$ and we measure each copy independently using the POVM $\cle$. 
Suppose we see each outcome $i$ a total of $N_i$ times such that $\sum_\inm N_i = N$. 
Then the normalized frequency vector $c := (N_i/N)_\inm$ is an unbiased estimator of the measurement probabilities $(\la E_i, \rho \ra)_\inm$. 
In the least-squares method, our goal is to find the Hermitian matrix that minimizes the following Euclidean loss: $\hrl := \argmin_{X \in \hrm(\sfa)} \| \mathrm{M}_\cle(X) - c \|^2_2 = (\me ^\dagger \me) \iv \me^\dagger(c).$
The inverse here requires $\cle$ to be informationally complete, so that $\{E_i\}_\inm$ spans $\hrm(\sfa)$.

\paragraph{Local Pauli measurements.}
We consider the special case where the POVM corresponds to local Pauli measurements on a multi-qubit system.
Let $\sfa \equiv \mathbb{C}^{d}$, for $d := 2^{n_\sfa}$, so that $\rho \in \dna$ is an $n_\sfa$-qubit state. 
We first consider the case where $n_\sfa = 1$. 
Then the Pauli POVM is the (normalized) collection of all eigenstates of the Pauli matrices: 
\begin{equation}
\begin{aligned}
    \clp_1 &:= \frac13 \left\{\kb{\pm, x},\kb{\pm, y}, \kb{\pm, z} \right\} \equiv \frac13\left\{\kb{o, s} : (o,s) \in \{+, -\} \times \{x, y, z\}\right\}.
\end{aligned}
\end{equation}
Here $s$ denotes the `setting' (the specific choice of the Pauli observable) and $o$ denotes the eigenvalue of the eigenvector that defines the measurement operator. 
The associated measurement channel then satisfies
\begin{equation}
\begin{aligned}
    ( \mathrm{M}^\dagger_{\mathcal{P}_1} \mps) (A) &= \frac1{3^2}\sum_{(o, s)} \kb{o, s}  \cdot \tr[\kb{o, s} \cdot A] = \frac1{3^2} (A + \tr[A] I ) = \frac1{3} \Delta_{2/3}(A),
\end{aligned}  
\end{equation}
where $\Delta_q(A) = (1-q) A + q \cdot \tr[A] I_2/2$ is the single-qubit depolarizing channel.
Moreover, $\Delta_{2/3} \iv (A) = 3 A - \tr[A] I_2 $, which allows us to write
\begin{equation}
\begin{aligned}
       \hrl = \left(\mathrm{M}^\dagger_{\mathcal{P}_1} \mps \right) \iv \left(\mathrm{M}^\dagger_{\mathcal{P}_1}(c)\right) &= 3 \left(3 \cdot \frac13\sum_{(o, s)} c_{(o,s)}\kb{o, s} - \frac1 3 \sum_{(o, s)} c_{(o, s)}I_2 \right) \\
        &= 3 \sum_{(o,s)} c_{(o,s)} \kb{o,s} - I_2,
\end{aligned}
\end{equation}
where the last step uses $\sum_{(o,s)} c_{(o,s)} = 1$.
For a general choice of $n_\sfa \equiv n$, the local Pauli POVM is simply the $n$-fold product of single-qubit Pauli POVMs. 
\begin{equation}
    \mathcal{P}_n := \bigtimes_{\inn} \mathcal{P}_1 = \frac1{3^n}\left\{ \kb{\vec{o}, \vec{s}} : \vec o \in \{+, -\}^n, \vec{s} \in \{x,y,z\}^n \right\},
\end{equation}
with $\ket{\vec{o}, \vec{s}} := \ket{o_1, s_1} \otimes \cdots \otimes \ket{o_n, s_n}$.
We then have
\begin{equation}
    ( \mathrm{M}^\dagger_{\mathcal{P}_n} \mpn)= \frac1{3^{n}} \Delta_{2/3}^{\otimes n} \quad  \Longrightarrow \quad      ( \mathrm{M}^\dagger_{\mathcal{P}_n} \mpn) \iv = 3^{n} (\Delta_{2/3} \iv )^{\otimes n},
\end{equation}
which allows us to write the LS estimate for $n$-qubit local Pauli measurement as
\begin{equation} \label{Eq:LSEstimate}
\begin{aligned}
        \hrl = \left(\mathrm{M}^\dagger_{\mathcal{P}_n} \mpn \right) \iv \left(\mathrm{M}^\dagger_{\mathcal{P}_n}(c)\right) &=   \sum_{(\vo, \vs)}  c_{(\vo, \vs)} \bigotimes_{i \in [n]} (3 \kb{o_i, s_i} - I_2).
\end{aligned}
\end{equation}

\paragraph{LS estimate approximates the true state.}
The key idea is that as the number of measurements $N$ approaches infinity, the LS estimate approaches the true state.
As shown in~\cite{guta2020fast}, using the matrix Bernstein inequality~\cite{tropp2015introduction}, {the probability of a deviation of any fixed size decays exponentially in $N$}:
\begin{equation} \label{Eq:LSConcentration}
    \mathrm{Pr}[\| \hrl - \rho\|_\infty \geq \eps] \leq d \exp\left(- \frac{3 N \eps^2}{8 \, g(d)}\right), \quad \eps \in [0, 1],
\end{equation}
where $\| \cdot\|_\infty$ is the spectral norm and $g(d) = 3^n = d^{\log_2 3}$ for $n = \log_2 d$ qubits measured in the local Pauli basis. {Throughout the remainder of this section, the \emph{concentration event} is the event $\{\|\hrl - \rho\|_\infty \leq \eps\}$, whose probability~\cref{Eq:LSConcentration} bounds from below by $1 - \delta$.}

\subsection{The projection step} \label{App:ProjectionStep}
The LS estimate is Hermitian and, for local Pauli measurements, of unit trace, but generally not PSD, thus it has to be \textit{projected} to the PSD cone, yet this need not be a true projection.
Existing works use different approaches to achieve this. \textcite{guta2020fast} use the Frobenius projection onto $\dna$ (mapping it to the nearest density matrix in Frobenius distance), computed by the algorithm of~\cite{smolin2012Efficient}; \textcite{surawy2022projected} generalize it to a \textit{thresholded} projection (their Algorithm 1, restated here as~\cref{Alg:SparseDensityEstimate}).
Briefly, the core idea behind thresholded projection is to choose a nonnegative threshold $\tau \geq 0$ and modify the spectrum $(\hat \lambda_i)_{i \in [d]}$  of $\hrl$ as follows:
\begin{equation}
    \hat \mu_i = \begin{cases}
        0 &: \hat \lambda_i \leq \tau, \\
        \hat \lambda_i + \tau &: \hat \lambda_i > \tau.
    \end{cases}
\end{equation}
{That is,} small and negative eigenvalues (noise) are annihilated while large eigenvalues (signal) are  boosted. 
Indeed, such a step promotes sparsity, which is typically desirable in QPT. 
Observe that the modified spectrum need not sum up to 1, and thus at the very next step the trace is restored, in one of two ways depending on the direction of the violation. 
If the filtered spectrum is \emph{supernormalized}, a common level is subtracted from every surviving eigenvalue until unit trace is restored. 
If it is \emph{subnormalized}, the shortfall is made up by walking down the raw LS spectrum and reinstating eigenvalues top-down until the trace reaches one, so directions eliminated by the filter can be brought back. Both branches are written out in~\cref{Alg:SparseDensityEstimate} to return a unit-trace estimate $\hr$.  
Eigenvalues are stated in \textit{descending} order, the convention that makes the threshold and cumulative-sum steps easiest to read.

\begin{algorithm}
\caption{Thresholded density estimation~\cite[Algorithm~1]{surawy2022projected}.}
\label{Alg:SparseDensityEstimate}

\DontPrintSemicolon
\SetKwProg{Fn}{Function}{:}{}
\SetKwFunction{FProjCP}{projCP}

\Fn{\FProjCP{${X}, \tau$}}{
    \tcp{$X$ Hermitian of unit trace; threshold $\tau \geq 0$}
    \tcp{1. Eigendecomposition and Sorting}
    $\boldsymbol{\lambda}, V \gets {\text{eigendecompose}(X)}$ \;
    
    Sort $\boldsymbol{\lambda}$ descending and permute columns of $V$ to match \;
    
    $\tau \gets$ input threshold \;
    
    \tcp{2. Apply threshold filter}
    $\boldsymbol{\mu} \gets \text{where}(\boldsymbol{\lambda} > \tau,\ \boldsymbol{\lambda} + \tau,\ 0)$\tcc*{zero-out small/negative, shift up large}
    
    \tcp{3. Setting to unit trace}
    \eIf{$\sum_i \mu_i \ge 1$}{
        \tcp{Scenario A: Supernormalized estimate (subtract mass from every $\mu_i > 0$.)}
        Find $x_0$ such that $\sum \max(\boldsymbol{\mu} - x_0, 0) = 1$ \;
        $\boldsymbol{\mu} \gets \max(\boldsymbol{\mu} - x_0, 0)$ \;
    }{
        \tcp{Scenario B: Subnormalized estimate (rescue raw eigenvalues top-down)}
        $\boldsymbol{C}_k \gets \sum_{i=1}^k \lambda_i $, with $\boldsymbol{C}_0 := 0$ \tcc*{$k$-Cumsum of LS eigenvalues; empty slices are empty.}
        Find smallest index $k$ such that $\boldsymbol{C}_k + k \cdot \tau \ge 1$ \;
        
        $\boldsymbol{\mu}[1 \dots k-1] \gets \boldsymbol{\lambda}[1 \dots k-1] + \tau$ \tcc*{Larger eigenvalues inflated by $\tau$}
        $\boldsymbol{\mu}[k] \gets 1 - (\boldsymbol{C}_{k-1} + (k-1) \cdot \tau)$ \tcc*{Partially inflate $k$-th eigenvalue}
        $\boldsymbol{\mu}[k+1 \dots d] \gets 0$ \tcc*{zero-out the rest}
    }
    
    \KwRet $V \text{diag}(\boldsymbol{\mu}) V^\dagger$ \;
}
\end{algorithm}

\cite[Lemma 2]{surawy2022projected} shows that the output $\hr$ of the above algorithm satisfies, for any threshold $\tau \leq \|\hrl - \rho\|_\infty$,
\begin{equation} \label{Eq:CP1InftyBound}
    \|\hr  - \rho\|_\infty \leq 2 \| \hrl - \rho \|_\infty.
\end{equation}
Two such choices that~{\textcite{surawy2022projected} suggest} are $\tau = 0$ and $\tau = -\lmin(\hrl)$,\footnote{We implement the rule as $\tau = \max(0, -\lmin(\hrl))$, since a negative threshold falls outside the stability guarantee of~\cite[Lemma~2]{surawy2022projected}; on every instance in this article the two rules coincide.} and thresholding with these values allows the translation of the sample complexity rates of the LS estimation.   
While the former choice of $\tau=0$ constitutes a true Frobenius projection and thus is contractive, it can typically return a high-rank estimate which is not desirable. 
Consequently,~{\textcite{surawy2022projected} use} the latter value for the threshold, which typically returns low-rank density estimates.

Note that different thresholds suit different regimes.
When the true state is full rank or only approximately low rank, $\tau = 0$ is the right choice as the estimate is exactly the Frobenius projection onto $\dna$ (\cref{Sec:ChannelRobustness}), and tends to have a high Choi rank.
When the true state has low rank $r$, one wants exactly $r$ eigenvalues to survive, and neither choice above achieves this as $N$ grows: at $\tau = 0$ every eigenvalue above the subtracted {level} survives, so the returned rank is set by how many noise eigenvalues happen to be positive and does not shrink with $N$.
The rule $\tau = -\lmin(\hrl)$ fares no better.
It discards every eigenvalue below the magnitude of the most negative one, so a noise eigenvalue survives exactly when the largest positive noise eigenvalue exceeds the largest negative one in magnitude.
The noise spectrum of the LS estimate is nearly symmetric about zero, so these two extremes are comparable at every $N$ and which of them is larger is close to a coin flip.
The rule therefore recovers the rank exactly in roughly half the trials at every $N$, with no trend (\cref{Fig:ThresholdStrategies}).

Since exact rank recovery buys the faster rate~(see \cref{Cor:DeterministicFastRate}), we want a rule that can do so.
To this end, we propose the \textit{Bernstein threshold}: $\tau = \beta_N$, with
\begin{equation} \label{Eq:BernsteinThreshold}
    \beta_N := \sqrt{\frac{8 g(d) \log(d/\delta)}{3N}}
\end{equation}
of~\cref{Eq:LSConcentration}, which corresponds to the radius of the \textit{concentration event} at $N$ samples for a given confidence $\delta$.\footnote{The tail bound~\cref{Eq:LSConcentration} is stated for deviations $\eps \in [0,1]$, so every statement we make at the radius $\beta_N$ presupposes $\beta_N \leq 1$, equivalently $N \geq \tfrac83 g(d)\log(d/\delta)$. 
This is no restriction in the regimes of interest.}
We use it because, unlike the two choices above, it captures the rank of the underlying state \emph{exactly} once $N$ is large enough: \cref{App:SupraNoise} shows that any threshold lying strictly between the noise level and the smallest nonzero eigenvalue does so, and \cref{App:BernsteinThreshold} shows that the Bernstein threshold lands in that interval on the concentration event.

\Cref{Fig:ThresholdStrategies} compares the four threshold strategies on two channels, the rank-$4$ amplitude damping channel at $\dab = 2^6$ and the rank-$2$ channel of~\cref{Fig:Headline} at $\dab = 2^8$.
At $\tau = 0$ the rank grows with $N$, from $6$ to $15$ of $64$ and from $7$ to $18$ of $256$ over the shot range; at $\tau = -\lambda_{\min}(\hrl)$ it fluctuates around the truth, between $4$ and $5$ and between $2$ and $3$, the rank being recovered in only $3$ to $9$ of the $10$ trials at every $N \geq 10^5$.
The Bernstein threshold under-counts before its empirical onset, admitting the eigenvalues one at a time as each rises above $\beta_N$, and is exact for every trial from $N = 3 \times 10^7$ on the damping channel and $10^6$ on the rank-$2$ channel, where the guaranteed onsets of~\cref{Cor:BernsteinRank}, $\beta_N < \lambda_r/2$, are $6 \times 10^7$ and $3 \times 10^6$; the halved threshold $\beta_N/2$ behaves in the same way with every step half a decade earlier, exact from $10^7$ and $3 \times 10^5$.
In accuracy, before its empirical onset the Bernstein threshold is the least accurate curve, since it discards true eigenvalues; past it the infidelity drops sharply and follows the $1/N$ rate of~\cref{Cor:DeterministicFastRate}, at fitted exponents $-1.01$ and $-0.99$, while $\tau = 0$ continues at $1/\sqrt N$, at $-0.50$ on both channels.
The two branches of $\tau = -\lambda_{\min}(\hrl)$ follow the two rates, the rank-recovered branch coinciding with the Bernstein curves and the rank-missed branch converging at $1/\sqrt N$ (hence one median per branch; see~\cref{Sec:AccuracyRate}).

\begin{figure}[t]
    \centering
    \includegraphics[width=\textwidth]{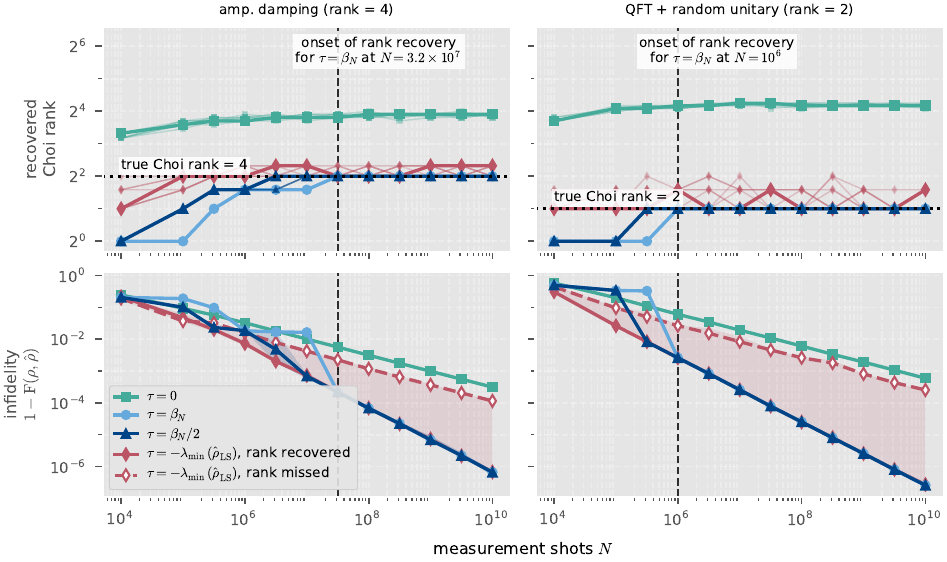}
    \caption{The threshold strategies of~\cref{Alg:SparseDensityEstimate} on two channels: independent amplitude damping on two of three qubits ($\dab = 2^6$, Choi rank $4$, $\lambda_4 = 0.031$; left) and the rank-$2$ channel of~\cref{Fig:Headline} ($\dab = 2^8$, $\lambda_2 = 0.48$; right).
    Every estimate is TP regularized by the fidelity projection; $10$ trials per shot count.
    Top row, recovered numerical Choi rank, translucent lines single trials and opaque lines their mode; bottom row, infidelity, medians with min--max bands, for $\tau = -\lambda_{\min}(\hrl)$ one median per branch, the trials that recovered the rank (filled markers) and those that did not (open markers).
    The dashed line marks the empirical onset of $\tau = \beta_N$, the smallest plotted shot count from which every trial returns the true rank.}
    \label{Fig:ThresholdStrategies}
\end{figure}

\subsection{Exact rank recovery at supra-noise thresholds} \label{App:SupraNoise}

In~\cite[Lemma 2]{surawy2022projected}, the bound~\cref{Eq:CP1InftyBound} is proved only for thresholds below the realized noise, $\tau \leq \|\hrl - \rho\|_\infty${,} as it suffices for their arguments.
Our density estimator sets $\tau$ at a multiple of the concentration radius, which on the concentration event sits \emph{above} the realized noise---precisely the regime that makes rank recovery certain, and precisely the regime~\cref{Eq:CP1InftyBound} does not cover.
The intuition is a gap: on the concentration event every null eigenvalue (which corresponds to $\li = 0$) of $\hrl$ lies within $\eps$ of $0$ and every signal eigenvalue (which corresponds to $\li > 0$) within $\eps$ of a value at least $\lambda_r$, so a threshold placed above $\eps$ but below $\lambda_r - \eps$ sits strictly between the noise band and the signal floor, and the filter separates the two without error.
The following lemma closes that gap: for \textit{supra-noise} thresholds, \cref{Alg:SparseDensityEstimate} recovers the rank exactly \emph{and} satisfies the stability bound $\|\hr - \rho\|_\infty \leq \tau + \eps \leq 2\tau$, the analogue of~\cref{Eq:CP1InftyBound} with the noise bound $\eps$ replaced by the threshold $\tau$, and the sharper $\|\hr - \rho\|_\infty \leq 2\eps$ on the interval where the rank is recovered.
\begin{tbox}
    \begin{lemma}[Stability at supra-noise thresholds] \label{Lem:SupraNoiseStability}
        Let $\rho \in \dna$ with $\rank(\rho) = r$ and smallest nonzero eigenvalue $\lambda_r$, and let $\hrl$ be the LS estimate which is Hermitian with unit trace and $\|\hrl - \rho\|_\infty \leq \eps$.
        Let $\hr$ be the output of~\cref{Alg:SparseDensityEstimate} at any threshold $\tau \geq \eps$. Then $\rank(\hr) \leq r$, and
        \begin{equation} \label{Eq:SupraNoiseStability}
            \|\hr - \hrl\|_\infty \leq \tau, \qquad \|\hr - \rho\|_\infty \leq \tau + \eps \leq 2 \tau .
        \end{equation}
        Moreover, if $\tau$ lies in the interval $[\eps, \lambda_r - \eps)$ then $\rank(\hr) = r$, $\|\hr - \hrl\|_\infty \leq \eps$ and $\|\hr - \rho\|_\infty \leq 2\eps$.
    \end{lemma}
\end{tbox}

\begin{proof}
    We first prove all claims under the condition $\tau \in [\eps, \lambda_r - \eps)$, and then drop its upper bound.
    Let $(\hli)_{i \in [d]}$ and $(\li)_{i \in [d]}$ be the decreasingly-ordered spectrum of $\hrl$ and $\rho$, respectively.
    By Weyl's inequality~\cite{bhatiamatrixanalysis}, $|\hli- \li| \leq  \|\hrl - \rho\|_\infty \leq \eps$ for every $i \in [d]$.
    The spectral-norm closeness implies that eigenvalues of the LS estimate $(\hli)_{i > r}$ are annihilated by the threshold:
    \begin{equation}
        i > r \Rightarrow \li = 0 \Rightarrow |\hli| \leq \eps \leq \tau \Rightarrow \hat\mu_i = 0. 
    \end{equation}
    For the \emph{leading eigenvalues} with index $i \leq r$, where too we have $|\li - \hli| \leq \eps$,
    \begin{equation}
       -\eps \leq  \hli - \li \Rightarrow \hli \geq \li - \eps  \geq \lambda_r - \eps > \tau. 
    \end{equation}
    Moreover, the sum of the shifted eigenvalues exceeds one:    
    \begin{equation}
        \sum_{i \in [r]} (\hli +  \tau) =\sum_{i \in [r]} \hli + r \tau   \geq  \sum_{i \in [r]} \li + r (\tau - \eps) \geq 1,
    \end{equation}    
    where the last inequality holds since $\tau \geq \eps$ and $\sum_{i \in [r]} \lambda_i = 1$.
    Consequently, the algorithm enters Scenario A, where the estimate is supernormalized.

    Let $h(x) := \sum_{i \in [r]} \max (\hli + \tau - x, 0)$. 
    Note that $h$ is continuous and strictly decreasing whenever it is positive, $h(0) \geq 1$, and $h(x) = 0$ for all $x \geq \max_i (\hli) + \tau$. 
    Consequently, there exists a unique $x_0$ where $h(x_0) = 1$.
    We now show that $|\tau - x_0| \leq \eps$ or equivalently $x_0 \in [\tau - \eps, \tau + \eps]$, by evaluating $h$ at the two extreme points of the above interval:
    \begin{equation}
        h(\tau - \eps) = \sum_{i \in [r]} \max (\hli + \eps, 0) \geq \sum_{i \in [r]} \li  = 1 \qaq h(\tau + \eps) = \sum_{i \in [r]} \max(\hli - \eps, 0) \leq \sum_{i \in [r]} \li = 1,
    \end{equation}
    where the inequality steps come from, via Weyl's inequality, $ |\hli - \li| \leq \eps$ for all $i \in [r]$. 
    Since $h$ continuously (strictly) decreases in this interval, we have $x_0 \in [\tau - \eps, \tau + \eps]$. 

    We now show that for all $i \in [r]$, $\hat \mu_i \equiv \hli + \tau - x_0 > 0$, which would imply all shifted eigenvalues survive and $\mathrm{rank}(\hr) = r$. 
    Here we denote $\{\hat \mu_i\}_{i \in [r]}$ to be the non-zero eigenvalues of $\hr$. 
    Indeed,
    \begin{equation}
        \hli + \tau - x_0 \geq \hli - \eps > \tau - \eps \geq 0,
    \end{equation}
    as required. 
    We now show stability of $\hr$. 
    Since $\hrl$ and $\hr$ share an eigenbasis, we have $\| \hrl - \hr\|_\infty = \max_{i \in [d]} | \hli -  \hat \mu_i|$.
    For any $i \in [r]$, we have $| \hli -  \hat \mu_i| = |\tau - x_0|  \leq \eps.$
    For $r < i \leq d$, $\li = 0$, so $|\hli - \hat\mu_i| = |\hli| \leq \eps$ by Weyl's inequality.
    Consequently, $\|\hr - \hrl\|_\infty \leq \eps \leq \tau$, which proves the first inequality of~\cref{Eq:SupraNoiseStability} and the bound $\|\hr - \hrl\|_\infty \leq \eps$ of the last claim.
    Using the triangle inequality then yields the second inequality, and on the interval the sharper $\|\hr - \rho\|_\infty \leq 2\eps$.

    We now drop the upper bound of the interval and assume only $\tau \geq \eps$.
    The first step is unchanged: every $\hli$ with $i > r$ is annihilated, so the survivors of the threshold filter, the indices with $\hli > \tau$, are among the first $m \leq r$ indices, while every non-survivor satisfies $-\tau \leq -\eps \leq \hli \leq \tau$.
    Without the upper-bound on $\tau$ some signal eigenvalues may be annihilated by a sufficiently large $\tau$.
    Thus either branch of step 3 can occur and we show that, regardless of which branch, $|\hat\mu_i - \hli| \leq \tau$ for all $i$ and at most $r$ of the $\hat\mu_i$ are nonzero.
    Thus $\|\hr - \hrl\|_\infty \leq \tau$ and $\rank(\hr) \leq r$ as before.
    
    We first consider the \textit{supernormalized} branch, where the shifted survivors sum beyond unity, $\sum_{i \leq m} (\hli + \tau) \geq 1$.
    The algorithm then finds the level $x_0 \geq 0$ at which $\sum_{i \leq m} \max(\hli + \tau - x_0, 0) = 1$ and returns $\hat\mu_i = \max(\hli + \tau - x_0, 0)$ for the survivors and $\hat\mu_i = 0$ for the non-survivors.
    Let $\mathcal{S} \subseteq [m]$ be the index set of survivors left positive by the mass-subtraction; $\mathcal{S} \neq \emptyset$ since the output has unit trace.
    Three facts.
    (a) $x_0$ is bounded, as summing $\hat\mu_i$ over $\mathcal{S}$ to one gives $\sum_{i \in \mathcal{S}} (\hli + \tau - x_0) = 1$, that is,
    \begin{equation}
        x_0 = \tau + \frac{\sum_{i \in \mathcal{S}} \hli - 1}{|\mathcal{S}|} .
    \end{equation}
    By Weyl's inequality $\hli \leq \li + \eps$, and $\sum_{i \in \mathcal{S}} \li \leq \tr [\rho] = 1$ for any index set, so $\sum_{i \in \mathcal{S}} \hli \leq 1 + |\mathcal{S}| \eps$ and hence $x_0 \leq \tau + \eps$.
    Together with $x_0 \geq 0$ this places $\tau - x_0$ in $[-\eps, \tau]$, so $|\tau - x_0| \leq \tau$ since $\eps \leq \tau$.
    (b) No survivor is zeroed.
    A survivor $i$ with $\hat\mu_i = 0$ would need $\hli + \tau \leq x_0 \leq \tau + \eps$, that is $\hli \leq \eps \leq \tau$, which contradicts $\hli > \tau$.
    Hence $\mathcal{S} = [m]$: every survivor keeps positive weight, and the nonzero $\hat\mu_i$ are exactly the $m \leq r$ survivors.
    (c) Every eigenvalue moves by at most $\tau$.
    For a survivor, $\hat\mu_i - \hli = \tau - x_0$, of size at most $\tau$ by (a); for a non-survivor, $\hat\mu_i - \hli = -\hli$, of size at most $\tau$ since $-\tau \leq \hli \leq \tau$.
    
    We now consider the \textit{subnormalized} branch, where the shifted survivors sum to sub-unity: $\sum_{i \leq m} (\hli + \tau) < 1$.
    The algorithm then adds mass to the eigenvalues in decreasing order.
    Let $C_k := \sum_{i \leq k} \hli$ be the $k$-cumulative-sum. 
    The algorithm takes the smallest index $k'$ with $C_{k'} + k'\tau \geq 1$, boosts the first $k' - 1$ eigenvalues by $\tau$, gives index $k'$ the remaining mass $\hat\mu_{k'} = 1 - C_{{k'}-1} - ({k'}-1)\tau$, and zeroes the rest.
    Observe that $m < k' \leq r$:
    for the upper-bound, we have 
    \begin{equation}
        C_r + r \tau = \sum_{i\in [r]} \hli  + r \tau \geq \sum_{i \in [r]} \lambda_i - r \eps + r \tau  = 1 + r(\tau - \eps) \geq 1. 
    \end{equation}
    So $r$ already satisfies the stopping rule, hence the earliest index must be at most $r$.
    The lower-bound $k' > m$ holds as $C_m + m \tau = \sum_{i \in [m]} (\hli + \tau) < 1$ is precisely the reason why the algorithm entered this branch.
    Since the survivors are the first $m$ indices and $k' > m$, index $k'$ and every later index is a non-survivor: $\hli \leq \tau$ for $i \geq k'$, and together with $\hli \geq \li - \eps \geq -\tau$ this gives $|\hli| \leq \tau$ there.
    Moreover $0 < \hat\mu_{k'} \leq \hat\lambda_{k'} + \tau$: the lower bound because the minimality of $k'$ implies $C_{k'-1} + (k'-1)\tau < 1$ and the upper bound because $C_{k'} + k'\tau \geq 1$.
    Now we compare corresponding eigenvalues $\hat \mu_i$ and $\hli$ to obtain upper-bound $\| \hr - \hrl \|_\infty$:
    \begin{equation}
        \hat \mu_i - \hli  \begin{cases}
            = \tau &: i < k' \quad (\hat\mu_i := \hli + \tau)\\
            \in [-\tau, \tau]    &: i = k'  \quad (\hat\mu_{k'} > 0,\ \hat\lambda_{k'} \leq \tau \text{ and } \hat\mu_{k'} \leq \hat\lambda_{k'} + \tau) \\
            = - \hli &: i > k'  \quad (\hat\mu_i := 0)\\
        \end{cases}.
    \end{equation}
    In all cases, $|\hat \mu_i - \hli|\leq \tau$ and thus $\|\hr - \hrl\|_\infty \leq \tau$.
    Moreover, $\rank(\hr) = k' \leq r$.

    In both branches the triangle inequality then gives $\|\hr - \rho\|_\infty \leq \tau + \eps \leq 2 \tau$.
\end{proof}

We close with a structural remark: on the interval $[\eps, \lambda_r - \eps)$ the threshold has no effect beyond selecting the survivors.
To see this, follow one signal eigenvalue $\hli$, $i \leq r$, through the algorithm.
The threshold step shifts it up to $\hli + \tau$, and the mass-subtraction step then lowers it by the level $x_0$ that restores unit trace, so its output value is $\hat\mu_i = \hli + \tau - x_0$.
When all $r$ signal eigenvalues survive and none is zeroed by the subtraction, the unit-trace condition $\sum_{i \leq r} (\hli + \tau - x_0) = 1$ can be solved for the level,
\begin{equation}
    x_0 = \tau + \frac{\sum_{j \leq r} \hat\lambda_j - 1}{r},
\end{equation}
and the threshold appears in it with the same sign as in the shift.
Substituting, the two occurrences of $\tau$ cancel exactly:
\begin{equation}
    \hat\mu_i = \hli + \tau - x_0 = \hli - \frac{\sum_{j \leq r} \hat\lambda_j - 1}{r}, \qquad i \in [r],
\end{equation}
and $\hat\mu_i = 0$ for $i > r$.
The output is thus the top-$r$ truncation of $\hrl$, renormalized by the uniform shift that returns its trace to one, and it does not depend on $\tau$ at all.
Any two thresholds in the interval therefore return the \emph{same} estimate, which is why the curves for $\tau = \beta_N$ and $\tau = \beta_N/2$ coincide once both have recovered the rank (\cref{Fig:ThresholdStrategies} and \cref{Sec:Numerics}).

\subsection{The Bernstein threshold recovers the rank} \label{App:BernsteinThreshold}
We now show that the Bernstein threshold $\tau = \beta_N$, which is set from $N${,} $d$ and $\delta$ alone, before the data are seen, inherits the guarantee of~\cref{Lem:SupraNoiseStability} {on the concentration event, which has probability at least $1 - \delta$, as soon as $N$ is large enough for it to land in the interval $[\eps, \lambda_r - \eps)$}.
\begin{tbox}
    \begin{cor}[Exact rank recovery at the Bernstein threshold] \label{Cor:BernsteinRank}
        Let $\rank(\rho) = r$ with smallest nonzero eigenvalue $\lambda_r$. Suppose we run~\cref{Alg:SparseDensityEstimate} at $\tau = \beta_N$.
        If $\beta_N < \lambda_r/2$, equivalently $N > \tfrac{32}{3} g(d) \log(d/\delta)/\lambda_r^2$, then with probability at least $1 - \delta$, the following hold:
        \begin{equation}
            \rank(\hr) = r \qaq \|\hr - \rho\|_\infty \leq 2\,\beta_N .
        \end{equation}
    \end{cor}
\end{tbox}
\begin{proof}
    On the event $\|\hrl - \rho\|_\infty \leq \beta_N$ of~\cref{Eq:LSConcentration}, which has probability at least $1 - \delta$, take $\eps = \beta_N$ in~\cref{Lem:SupraNoiseStability}: $\tau = \beta_N$ is the left end of the interval $[\eps, \lambda_r - \eps)$, and $\beta_N < \lambda_r/2$ is exactly $\tau < \lambda_r - \beta_N$, so the lemma applies with $\tau + \beta_N = 2\beta_N$.
    The condition on $N$ is $\beta_N^2 = \tfrac83 g(d)\log(d/\delta)/N$ inserted into $\beta_N < \lambda_r/2$.
\end{proof}
The threshold therefore captures the rank on the concentration event beyond an explicit onset, that is, with probability at least $1 - \delta$ at every such $N$.
Below the onset it may underfit, returning fewer than $r$ eigenvalues, which is the behavior visible at the smallest shot counts in~\cref{Fig:Headline,Fig:ChannelRobustness}.
Numerically this radius is conservative --- across the instances of~\cref{Fig:AccuracyVsN,Fig:ChannelRobustness} the realized noise $\|\hrl - \rho\|_\infty$ sits at $0.2$ to $0.29$ of $\beta_N$ --- so a threshold well below $\beta_N$ still recovers the rank in practice, and the numerical experiments of~\cref{Sec:Numerics} accordingly use $\tau = \beta_N/2$, though we do not prove it here.

\subsection{Rates in fidelity: linear and quadratic rank conversion} \label{App:Rates}
We now aim to upper-bound the infidelity of the thresholded estimate.
For any $\rho, \sigma \in \dna$ with $\rank(\rho) \leq r$, we have
\begin{equation} \label{Eq:RankConversion}
    1 - \rf(\rho, \sigma) \leq \tfrac12 \|\rho - \sigma\|_1 \leq r \|\rho - \sigma\|_\infty,
\end{equation}
where the first inequality follows from the Fuchs--van de Graaf inequalities~(\eqref{Eq:FuchsVan}) and the second inequality is \cite[Lemma 2]{guta2020fast}.

Combining~\cref{Eq:LSConcentration}, \cref{Eq:CP1InftyBound} and~\cref{Eq:RankConversion} recovers Theorem~1 of~\textcite{guta2020fast}: for any $\rho \in \dna$ with $\rank(\rho) \leq r$, and at either of {the} PLS-QST thresholds, the thresholded estimate satisfies $\rf(\rho, \hr) \geq 1 - \eps$ with probability at least $1 - \delta$ as soon as the number of copies $N$ is at least:
\begin{equation} \label{Eq:GutaQST}
    N \geq \tfrac{32}{3}\, r^2 g(d) \log(d/\delta)/\eps^2 = O\Bigl(\frac{r^2 g(d)}{\eps^2} \log(d/\delta)\Bigr), \qquad g(d) = d^{\log_2 3} \leq d^{1.6} .
\end{equation}
\Cref{Eq:CP1InftyBound} {does} not hold at the Bernstein threshold since $\tau = \beta_N$ exceeds the noise on the concentration event.
\Cref{Lem:SupraNoiseStability} replaces it: at $\tau = \eps = \beta_N$ on the concentration event, its stability guarantee gives $\|\hr - \rho\|_\infty \leq 2\beta_N$ at every $N$ with $\beta_N \leq 1$ and with no onset, and~\cref{Eq:RankConversion} turns this into $1 - \rf(\rho, \hr) \leq 2r\beta_N$, the same bound Guţă et al.\ obtain at the sub-noise thresholds, so the Bernstein threshold inherits the slow rate~\cref{Eq:GutaQST} without a rank bound and never returns more than $r$ eigenvalues.
What the onset buys is the other direction, $\rank(\hr) \geq r$, through~\cref{Cor:BernsteinRank}.
\begin{tbox}
    \begin{theorem}[Bernstein-threshold QST] \label{Thm:SparseQST}
        Let $\rho \in \dna$ with $\rank(\rho) = r$ and smallest nonzero eigenvalue $\lambda_r > 0$, let $\delta \in (0,1)$ and $\eps \in (0, 1]$, and let $\hr$ be the output of~\cref{Alg:SparseDensityEstimate} at the Bernstein threshold $\tau = \beta_N$, for the sample size of~\cref{Eq:GutaQST}.
        Then
        \begin{equation} \label{Eq:BernsteinThresholdBound}
            \beta_N = \sqrt{\frac{8 g(d) \log(d/\delta)}{3N}} \leq \frac{\eps}{2 r},
        \end{equation}
        and with probability at least $1 - \delta$ the following hold.
        \begin{enumerate}
            \item \emph{Slow rate.} For every such $\eps$,
            \begin{equation} \label{Eq:BernsteinQST}
                \rank(\hr) \leq r \qaq \rf(\rho, \hr) \geq 1 - \eps.
            \end{equation}
            \item \emph{Rank certificate.} If moreover $\eps < r \lambda_r$, then $\beta_N < \lambda_r/2$ and $\rank(\hr) = r$.
        \end{enumerate}
    \end{theorem}
\end{tbox}
\begin{proof}
    \Cref{Eq:BernsteinThresholdBound} follows from substituting~\cref{Eq:GutaQST} in the definition~\cref{Eq:BernsteinThreshold} of $\beta_N$; in particular $\beta_N \leq 1$, so the tail bound~\cref{Eq:LSConcentration} applies and the concentration event $\|\hrl - \rho\|_\infty \leq \beta_N$ has probability at least $1 - \delta$.
    On that event, \cref{Lem:SupraNoiseStability} with $\tau = \eps = \beta_N$ gives $\rank(\hr) \leq r$ and $\|\hr - \rho\|_\infty \leq 2\beta_N$, and~\cref{Eq:RankConversion} with $\sigma = \hr$ then gives
    \begin{equation}
        1 - \rf(\rho, \hr) \leq r \| \hr - \rho\|_\infty \leq 2 r \beta_N \leq \eps,
    \end{equation}
    which is (1).
    For (2), $\eps < r\lambda_r$ gives $\beta_N \leq \eps/(2r) < \lambda_r/2$, so~\cref{Cor:BernsteinRank} applies on the same event and gives $\rank(\hr) = r$.
\end{proof}
Part (1) reproduces the sample complexity~\cref{Eq:GutaQST} of~\textcite{guta2020fast} without a rank bound as input and with no condition on the spectrum.
The proviso $\eps < r\lambda_r$ enters only through the rank certificate of part (2), which plays no role in the bound of part (1); it is what~\cref{Prop:QuadraticConversion} exploits in~\cref{Cor:DeterministicFastRate} to improve the rate from $1/\eps^2$ to $1/\eps$.
The $\eps$-dependence of~\cref{Thm:SparseQST} is decided by \cref{Eq:RankConversion}, which converts the operator-norm error $ \Delta \equiv \|\hr - \rho\|_\infty$ into infidelity \emph{linearly}, $1 - \rf \leq r \Delta$.
This linear conversion is tight only when the estimate puts weight where the true state has none.
For example, take $\rho$ pure and $\si = (1-t)\rho + t \si_\perp$ with $\si_\perp$ orthogonal to $\rho$. 
Then $\rf(\rho, \si) = \sqrt{1-t}$, so $1 - \rf \approx t/2$ is of \emph{first} order in the error, whatever the state $\si_\perp$.
The following proposition is the replacement for~\cref{Eq:RankConversion} that shows that the dependence can be made quadratic if no rank-mismatch occurs.
The tomographic consequence is~\cref{Cor:DeterministicFastRate} below.
\begin{tbox}
    \begin{prop}[Quadratic rank conversion] \label{Prop:QuadraticConversion}
        Let $\rho \in \dna$ have rank $r$, with smallest \emph{nonzero} eigenvalue $\lambda_r > 0$, and let $\si \in \dna$ satisfy $\rank(\si) \leq r$ and $\rsin < \lambda_r$. Then
        \begin{equation} \label{Eq:QuadraticConversion}
            1 - \rf(\rho, \si)  \leq  \frac{r \rsin^2}{\lambda_r - \rsin}.
        \end{equation}
    Under the condition $\rsin \leq \lambda_r/4$, the inequality can be further upper-bounded as follows:
    \begin{equation}
        1 - \rf(\rho, \si)  \leq  \frac{r \rsin^2}{\lambda_r - \rsin} \leq \frac{4 r \rsin^2}{3 \lambda_r}
    \end{equation}
    \end{prop}
\end{tbox}
\begin{proof}
    We first establish that the spectral-norm proximity ($\rsin < \lambda_r$) forces $\rank(\sigma) = r$. We will prove this by contradiction.
    Let us assume that $\rank(\sigma)< r$, which implies $\mathrm{dim}(\ker \sigma) \geq d - r + 1$.
    Moreover, $\rank(\rho) = \mathrm{dim}(\operatorname{supp} \rho) = r$, so $\mathrm{dim}(\ker \sigma) + \mathrm{dim}(\operatorname{supp} \rho) > d$, which implies the existence of a unit vector $v \in \ker \sigma \cap \operatorname{supp} \rho$; for this $v$,
    \begin{equation} \label{Eq:KernelSupportVector}
       \rsin \geq \la v, (\rho - \sigma) v \ra = \la v, \rho v \ra \geq  \lambda_r,
    \end{equation}
    which contradicts $\rsin < \lambda_r$.
    Thus $\rank(\sigma) = r$.
    
    Recall that for any pair of states, the variational form of squared Bures distance implies 
    \begin{equation} \label{Eq:BuresVariational}
        1 - \frs = \min_{U \in \mathrm{U}(d)} \frac12  \| \rho \hf - \sigma \hf U\|_2^2  \leq \frac12 \| \rho \hf - \sigma \hf \|_2^2. 
    \end{equation}
    We now aim to further upper-bound the above inequalities in terms of the Frobenius distance between $\rho$ and $\sigma$.
    To this end, let $\rho = \sum_i a_i \kb{u_i}$ and $\si = \sum_j b_j \kb{v_j}$ be the spectral decomposition over \emph{complete} orthonormal bases so $a_i, b_j \geq 0$. For every pair,
    \begin{equation} \label{Eq:DividedDifferences}
        \la u_i | \sqrt\rho - \sqrt\si | v_j \ra
        = \Lt(\sqrt{a_i} - \sqrt{b_j}\Rt) \la u_i , v_j \ra
        = \frac{\la u_i , (\rho - \si) v_j \ra}{\sqrt{a_i} + \sqrt{b_j}},
    \end{equation}
    whenever $a_i + b_j > 0$, while pairs with $a_i = b_j = 0$ evaluate to zero on both sides.
    Assuming the eigenvalues to be ordered decreasingly, we have $a_i \geq \lambda_r$ for all $i \in [r]$. 
    By Weyl's inequality, $ - \rsin \leq b_j - a_j \leq \rsin$, which implies $b_j \geq a_j - \rsin \geq \lambda_r - \rsin$ for all $j \in [r]$.
    Observe that
    \begin{equation} \label{Eq:49}
        (\sqrt{a_i} + \sqrt{b_j})^2 \geq a_i + b_j \geq \lambda_r + \lambda_r - \rsin \geq \lambda_r - \rsin \quad i,j \in [r]. 
    \end{equation}
    The same bound holds when exactly one of $a_i, b_j$ is nonzero, since then $(\sqrt{a_i} + \sqrt{b_j})^2$ equals that eigenvalue, which is at least $\lambda_r - \rsin$ by the two bounds above; the nonzero one is necessarily indexed in $[r]$ because $\rank(\rho) = \rank(\sigma) = r$.
    Hence $(\sqrt{a_i} + \sqrt{b_j})^2 \geq \lambda_r - \rsin$ for every pair with $a_i + b_j > 0$.
    Let $\mathcal{S} := \{(i,j) \in [d]^2 : a_i + b_j > 0\} $ be the set of such index pairs.
    We thus have
    \begin{equation} \label{Eq:RankRecoveryChain}
    \begin{aligned}
        \| \rho \hf - \si \hf\|_2^2 =  \sum_{i, j \in [d]} |\la u_i, (\rho \hf - \si \hf) v_j \ra |^2 &= {\sum_{(i,j) \in \mathcal{S}}} \frac{|\la u_i, (\rho - \sigma) v_j \ra|^2}{(\sqrt{a_i} + \sqrt{b_j})^2} \\ 
        &\leq {\sum_{(i,j) \in \mathcal{S}}} \frac{|\la u_i, (\rho - \sigma) v_j \ra|^2}{ \lambda_r - \rsin} \\ &\leq {\sum_{i,j \in [d]}} \frac{|\la u_i, (\rho - \sigma) v_j \ra|^2}{ \lambda_r - \rsin} =\frac{\| \rho - \sigma\|_2^2}{\lambda_r - \rsin},
    \end{aligned}
    \end{equation}
    where the first and last equalities use the fact that $\|A\|_2^2 = \tr[AA \dg] = \sum_{i,j} |\la u_i, Av_j\ra|^2$ for any $A \in \mathrm{L}(\sfa)$, the second equality uses~\cref{Eq:DividedDifferences} on $\mathcal{S}$ together with the fact that the pairs outside $\mathcal{S}$, where $a_i = b_j = 0$, contribute zero to the left-hand side, and the first inequality is a consequence of~\cref{Eq:49}.
    Finally, observe that $\rank(\rho - \sigma) \leq 2 r$, which implies $\| \rho - \sigma\|_2^2 \leq 2 r \rsin^2$.
    Combining this with~\cref{Eq:RankRecoveryChain,Eq:BuresVariational} results in  
    \begin{equation}
        1 - \frs \leq \Lt(\frac12\Rt) \frac{2 r \rsin^2}{\lambda_r - \rsin} \leq  \frac{4 r \rsin^2}{3\lambda_r},
    \end{equation}
    with the last inequality holding if $\rsin \leq \lambda_r/4$. 
    This concludes the proof.
\end{proof}

\begin{remark}[Off-support weight sets the rate] \label{Rem:OffSupportWeight}
    The rank hypothesis of~\cref{Prop:QuadraticConversion} is sufficient but not necessary: what decides the rate is the weight $\ell$ that $\sigma$ places outside $\operatorname{supp}(\rho)$,
    \begin{equation} \label{Eq:OffSupportMass}
        \ell := \tr[(I - P_\rho) \sigma] = 1 - \la P_\rho, \sigma \ra,
    \end{equation}
    where $P_\rho$ is the projector onto $\operatorname{supp}(\rho)$.
    Consider the pinching channel $\mathcal{P}(X) := P_\rho X P_\rho + (I - P_\rho) X (I - P_\rho)$, which leaves $\rho$ invariant.
    The data-processing inequality for the fidelity gives
    \begin{equation} \label{Eq:OffSupportLower}
        \rf(\rho, \sigma) \leq \rf\bigl(\mathcal{P}(\rho), \mathcal{P}(\sigma)\bigr) = \rf(\rho, P_\rho \sigma P_\rho) \leq \sqrt{\tr[\rho]\, \tr[P_\rho \sigma P_\rho]} = \sqrt{1 - \ell},
    \end{equation}
    where the equality holds because the block $(I - P_\rho) \sigma (I - P_\rho)$ of $\mathcal{P}(\sigma)$ is orthogonal to $\rho$, and the last inequality is $\rf(A, B) \leq \sqrt{\tr[A] \tr[B]}$ for $A, B \geq 0$.
    Hence $1 - \rf(\rho, \sigma) \geq 1 - \sqrt{1 - \ell} \geq \ell/2$, and a rate $O(\Delta^2)$ requires $\ell = O(\Delta^2)$.
    Now, if $\sigma = (1 - \eta)\, \sigma_r + \eta\, \omega$ with $\sigma_r$ as in~\cref{Prop:QuadraticConversion}, $\Delta := \|\rho - \sigma_r\|_\infty \leq \lambda_r/4$ and $\omega \in \dna$ arbitrary, then $\sigma \geq (1 - \eta)\, \sigma_r$ gives $\rf(\rho, \sigma) \geq \sqrt{1 - \eta}\, \rf(\rho, \sigma_r) \geq \rf(\rho, \sigma_r) - \eta$, that is,
    \begin{equation} \label{Eq:LeakyConversion}
        1 - \rf(\rho, \sigma) \leq \eta + \frac{4 r \Delta^2}{3 \lambda_r},
    \end{equation}
    so an estimate of full rank still attains the fast rate when $\eta = O(\Delta^2)$.
    Exact rank recovery (\cref{Cor:BernsteinRank}) is the simplest way to meet this condition.
    Numerical experiments indicate that the thresholds $\tau = 0$ and $\tau = -\lmin$ fail to do so.
\end{remark}

We may now derive the fast rate ($1 - \rf \propto N \iv$) by combining~\cref{Cor:BernsteinRank} with~\cref{Prop:QuadraticConversion}: past a finite onset, the Bernstein threshold recovers the rank with probability at least $1 - \delta$, and the quadratic conversion then applies.

\begin{tbox}
    \begin{cor}[Fast rate at the Bernstein threshold] \label{Cor:DeterministicFastRate}
        Let $\rho \in \dna$ with $\rank(\rho) = r$ have smallest nonzero eigenvalue $\lambda_r > 0$, let $\delta \in (0,1)$, and run~\cref{Alg:SparseDensityEstimate} at the Bernstein threshold $\tau = \beta_N$ of~\cref{Eq:BernsteinThreshold}.
        If
        \begin{equation} \label{Eq:FastRateOnset}
            N \geq \frac{512}{3}\, \frac{g(d) \log(d/\delta)}{\lambda_r^2},
            \qquad \text{equivalently} \qquad \beta_N \leq \frac{\lambda_r}{8},
        \end{equation}
        then with probability at least $1 - \delta$,
        \begin{equation} \label{Eq:DeterministicFastRate}
            \rank(\hr) = r
            \qaq
            1 - \rf(\rho, \hr) = O\Bigl(\frac{r\, \beta_N^2}{\lambda_r}\Bigr) = O\Bigl(\frac{r\, g(d) \log(d/\delta)}{\lambda_r N}\Bigr).
        \end{equation}
    \end{cor}
\end{tbox}
Explicitly, the bound reads $1 - \rf(\rho, \hr) \leq \tfrac{16}{3} r \beta_N^2/\lambda_r = \tfrac{128}{9} r g(d)\log(d/\delta)/(\lambda_r N)$.

\begin{proof}
    Squaring the definition~\cref{Eq:BernsteinThreshold} of $\beta_N$ shows that the two conditions in~\cref{Eq:FastRateOnset} are equivalent.
    Since $\beta_N \leq \lambda_r/8 < \lambda_r/2$, \cref{Cor:BernsteinRank} applies: with probability at least $1 - \delta$,
    \begin{equation}
        \rank(\hr) = r \qaq \|\hr - \rho\|_\infty \leq 2\beta_N \leq \frac{\lambda_r}{4} .
    \end{equation}
    On this event, $\sigma = \hr$ satisfies the hypotheses of~\cref{Prop:QuadraticConversion}, namely $\rank(\sigma) \leq r$ and $\rsin \leq \lambda_r/4 < \lambda_r$, so its second bound gives
    \begin{equation}
        1 - \rf(\rho, \hr) \leq \frac{4 r \|\hr - \rho\|_\infty^2}{3 \lambda_r} \leq \frac{4 r (2\beta_N)^2}{3 \lambda_r} = \frac{16}{3}\, \frac{r \beta_N^2}{\lambda_r} .
    \end{equation}
    Inserting $\beta_N^2 = \tfrac83 g(d) \log(d/\delta)/N$ gives~\cref{Eq:DeterministicFastRate}.
\end{proof}

\subsection{Proof of the sample complexity of FPLS-QPT} \label{App:QPTProof}
Everything is now in place to prove the main non-adaptive guarantee: the lifting theorem (\cref{Thm:Main}) reduces it to a QST guarantee for the Bernstein-threshold density estimate, the stability lemma (\cref{Lem:SupraNoiseStability}) and the linear conversion~\cref{Eq:RankConversion} supply the slow rate at every $N$, and the fast rate beyond the onset is~\cref{Cor:DeterministicFastRate}.
We restate the theorem for convenience.
\begin{tbox}
    \NonAdaptQPTThm*
\end{tbox}
\begin{proof}
    By~\cref{Thm:Main}, it suffices that the density estimate $\hr$ of step 3 in \cref{Alg:FPLSNonAdaptive} achieves accuracy $\eps/4$, and by~\cref{Lem:AutoInvertible} the condition $\eps/4 < 1/(2d_\sfa)$ then makes its marginal invertible, so partial normalization in step 4 in \cref{Alg:FPLSNonAdaptive} is well defined.
    For the slow bound~\cref{Eq:SlowQPTSamples}, \cref{Lem:SupraNoiseStability} gives $\|\hr - \cphi\|_\infty \leq 2\beta_N$ with probability at least $1 - \delta$ at every $N$, and~\cref{Eq:RankConversion} turns this into $1 - \rf(\cphi, \hr) \leq 2r\beta_N$; requiring $2r\beta_N \leq \eps/4$ is, by the definition~\cref{Eq:BernsteinThreshold} of $\beta_N$, the explicit sample size $N \geq \tfrac{512}{3} r^2 \dimab^{1.6}\log(\dimab/\delta)/\eps^2$.
    For the fast bound~\cref{Eq:BernsteinQPTSamples}, \cref{Cor:DeterministicFastRate} gives $\rank(\hr) = r$ and $1 - \rf(\cphi, \hr) \leq \tfrac{16}{3} r \beta_N^2/\lambda_r$ with probability at least $1 - \delta$ whenever $\beta_N \leq \lambda_r/8$.
    The explicit sample size is the condition $\beta_N^2 \leq 3\lambda_r \eps/(64 r)$, under which the infidelity bound is at most $\eps/4$; for $\eps \leq r\lambda_r/3$ this condition implies $\beta_N^2 \leq \lambda_r^2/64$, so the onset $\beta_N \leq \lambda_r/8$ of~\cref{Cor:DeterministicFastRate} is then automatic; this is why the theorem needs no separate condition on $N$ beyond the sample size.
    Finally the fidelity projection is a congruence by the invertible operator $\hr_\sfa^{-1/2} \ot I_\sfb$, so $\fc(\hat\Phi)$ has the same rank $r$ as $\hr$.
\end{proof}

\section{Calibrated benchmark instances} \label{App:SyntheticInstances}
This appendix discusses how data in~\cref{Sec:Numerics} (\cref{Fig:HIPPerformanceAnalysis}) is generated for studying HIP iterations {against} Choi-state dimension. 
The naïve approach would be to generate density estimates (based on the $-\lmin$-rule) from a fixed number of shots at different $\dab$ and numerically estimate the number of iterations required for convergence. 
However, this {approach is flawed}, as a fixed number of measurement shots would make estimates more accurate at low dimensions and less accurate at high dimensions, and thereby make it seem that the iterations increase with dimension.
What we ideally want is to initialize every density estimate at a fixed distance from $\cab$ and then see how the HIP iterations for convergence {scale} with dimension. 

To this end, we note that for any $\rho \in \dnab$ with $\rho_\sfa > 0$, its fidelity to (the nearest point in) $\cab$ is 
\begin{equation} \label{Eq:DistanceToCPTP}
    \mathrm{F}(\rho, \pc(\rho)) = \frac{1}{\sqrt{d_\sfa}} \tr[\sqrt{\rho_\sfa}],
\end{equation}
which depends solely on the marginal $\rho_\sfa$, so that $\dpu(\rho, \cab) = \sqrt{1 - \mathrm{F}(\rho, \pc(\rho))^2}$ costs one $d_\sfa \times d_\sfa$ eigendecomposition.\footnote{The choice of purified distance here is {arbitrary}. Any distance based on fidelity would suffice.}

We thus design our simulation strategy as follows. 
We fix the required fidelity (\cref{Fig:HIPPerformanceAnalysis} uses $\mathrm{F} = {0.9987}, 0.995, 0.980, 0.954$, i.e.\ $\dpu = 0.05, 0.1, 0.2, 0.3$) and a true channel of appropriate Choi rank and -dimension, randomly sampled from the BCSZ ensemble~\cite{bruzda2009random}: a Haar-random isometry for panel (a) and Choi ranks $\log_2 \dimab$ and $\dimab$ for panels (b) and (c). 
We want to find the measurement count $N$ such that the PLS density estimate (using the $\tau = -\lmin$ rule) sits at the predetermined fidelity from $\cab$. 
To this end, we iteratively vary $N$ (increase (decrease) if estimated fidelity is lower (higher) than the required fidelity) until the estimated fidelity is within 1\% of the target or 30 iterations are completed; since the noise is random, the realized distance is only statistically monotone in $N$, and the $1\%$ band is met by about a fifth of the instances, the rest stopping at the cap. 
This density estimate is then fed to the HIP algorithm to numerically estimate the number of iterations required for convergence (\cref{Fig:HIPPerformanceAnalysis}). 

The following table lists, for the isometry family of panel (a), the calibrated shot count $N$ and the realized $\dpu$ at each target, as medians over the ten instances per dimension. 
\begin{center}
    \footnotesize
    \begin{tabular}{rcccccccc}
        & \multicolumn{2}{c}{$\dpu = 0.05$} & \multicolumn{2}{c}{$\dpu = 0.1$} & \multicolumn{2}{c}{$\dpu = 0.2$} & \multicolumn{2}{c}{$\dpu = 0.3$} \\
        $\dimab$ & $N$ & realized & $N$ & realized & $N$ & realized & $N$ & realized \\
        \hline
        $2^2$    & $6.1 \times 10^2$ & $0.055$ & $1.4 \times 10^2$ & $0.106$ & $3.8 \times 10^1$ & $0.217$ & $2.3 \times 10^1$ & $0.331$ \\
        $2^4$    & $7.8 \times 10^3$ & $0.050$ & $1.7 \times 10^3$ & $0.101$ & $3.0 \times 10^2$ & $0.206$ & $9.9 \times 10^1$ & $0.315$ \\
        $2^6$    & $5.6 \times 10^4$ & $0.051$ & $1.4 \times 10^4$ & $0.100$ & $3.5 \times 10^3$ & $0.203$ & $9.6 \times 10^2$ & $0.331$ \\
        $2^8$    & $3.6 \times 10^5$ & $0.051$ & $8.6 \times 10^4$ & $0.100$ & $2.0 \times 10^4$ & $0.201$ & $8.5 \times 10^3$ & $0.298$ \\
        $2^{10}$ & $2.6 \times 10^6$ & $0.050$ & $5.5 \times 10^5$ & $0.102$ & $1.1 \times 10^5$ & $0.201$ & $3.5 \times 10^4$ & $0.326$ \\
    \end{tabular}
\end{center}

The medians track the target to within $10\%$ at every level, with the largest misses at the smallest dimensions and at $\dpu = 0.3$, where $\dpu$ is a coarser function of $N$; on the two random-channel families the deviations reach $20\%$.
The instances are therefore approximately, not exactly, distance matched, and dimension is the main but not the only quantity that varies along each curve of~\cref{Fig:HIPPerformanceAnalysis}; the fitted exponents of~\cref{Eq:HIPIterScaling} should be read with this residual mismatch in mind.
Instances at a given dimension and family share the true channel and differ only in the measurement randomness.
On the $140$ calibrated instances at $\dimab$ from $2^2$ to $2^8$, the ratio of the infidelities after and before the fidelity projection never exceeded $1$ (largest value $0.995$, median $0.91$), against the worst case of $4$ in~\cref{Thm:Main}; \cref{App:ContractionNumerics} extends this test to more channel families and to $\dimab = 2^{10}$.
A full-rank channel needs a far better estimate to reach the same fidelity to $\cab$, so its shot counts are larger, by $270\times$ at $\dimab = 2^{10}$ and $\dpu = 0.05$.
Finally, note that setting $N$ from the bound~\cref{Eq:SlowQPTSamples} at a fixed target $\eps$ instead would not help: at $\eps = 0.1$, i.e.\ a target $\dpu = 0.44$, the isometry instances realize $\dpu = 1.6 \times 10^{-3}$ at $\dimab = 2^2$ and $0.8 \times 10^{-3}$ at $2^{10}$ (medians over three trials), two orders of magnitude inside the target and halving across the sweep, because the slack in the bound is both large and dimension dependent.

\paragraph{Memory constraints due to the simulation {pipeline}.} 
The least-squares stage materializes the $6^n$ outcome distribution, giving a peak footprint of roughly $64 \cdot 6^{n}$ bytes --- about $3.6$ GiB at $n = 10$, which is what sets $\dimab = 2^{10}$ as the largest dimension reported here rather than any property of the projection itself. 
The fidelity projection costs microseconds at that dimension (\cref{Fig:Headline}, left) and would run far beyond it. 

\section{Additional proofs and numerics} \label{App:AdditionalProofs}

\subsection{Proof of \cref{Thm:FactorizedImplementation}} \label{App:FactorizedProof}

For better readability, let us first restate \cref{Thm:FactorizedImplementation}.
\begin{tbox}
\fidelityprojectioninKrausform*    
\end{tbox}

\begin{proof}
    Let $\ket\omega := d_\sfa^{-1/2} \sum_a \ket{a}\ket{a}$ and, for $M \in \bbc^{d_\sfb \times d_\sfa}$, let
    \begin{equation}
        \vecm(M) := (I_\sfa \ot M)\ket\omega = d_\sfa^{-1/2} \sum_{a, b} M_{ba} \ket{a}\ket{b},
    \end{equation}
    a linear bijection $\bbc^{d_\sfb \times d_\sfa} \to \bbc^{\dimab}$ that stores the entries of $M$ as the components of a bipartite vector; then $\fc(\Theta) = \sum_i \vecm(K_i)\vecm(K_i)\dg$~\cite[Ch.~2]{Watrous2018Theory}.
    Conversely, let us write each eigenvector in the spectral decomposition $\hr = \sum_i \lambda_i \kb{v_i}$ as a bipartite vector, $\ket{v_i} = \sum_{a, b} (V_i)_{ab} \ket{a}\ket{b}$ with $V_i \in \bbc^{d_\sfa \times d_\sfb}$ the array of its components; then $\sqrt{\lambda_i}\, \ket{v_i} = \vecm(K_i)$ for $K_i := \sqrt{d_\sfa \lambda_i}\, V_i^{\mathsf T}$, so the eigendecomposition furnishes Kraus operators directly, by a reshape and a rescaling.
    Tracing out $\sfb$ term by term gives
    \begin{equation} \label{Eq:MarginalIsDefect}
        d_\sfa\, \hr_\sfa  =  \Big(\sum_i K_i \dg K_i\Big)^{\mathsf T}  =  R^{\mathsf T},
    \end{equation}
    so $\Theta$ is trace preserving exactly when $R = I_\sfa$.
    Since $(M \ot I_\sfb)\vecm(K_i) = \vecm(K_i M^{\mathsf T})$, conjugating a Choi state by $M \ot I_\sfb$ multiplies each Kraus operator on the right by $M^{\mathsf T}$; with $M = \hr_\sfa\ihf$ in~\cref{Eq:ClosedFormProj} and~\cref{Eq:MarginalIsDefect}, the normalization is $K_i \mapsto K_i R\ihf$, which is trace preserving identically, $\sum_i {K_i'}\dg K_i' = R\ihf R\, R\ihf = I_\sfa$.
    The eigenpairs are supplied by the density estimator, which thresholds a spectrum, and $\vecm^{-1}$ is a reshape, so the $K_i$ are available from the supplied eigenpairs with $O(r \dimab)$ additional work; forming $R$ and applying $R\ihf$ to the $r$ Kraus operators cost $O(r d_\sfa^2 d_\sfb)$ each, diagonalizing $R$ costs $O(d_\sfa^3)$, and the memory is that of the $K_i$ and of $R$.
\end{proof}

\subsection{Marginal invertibility} \label{App:Invertibility}
This appendix proves the implication used in~\cref{Rem:Singular}: sufficient accuracy of the density estimate forces its $\sfa$-marginal to be invertible, so the fidelity projection is well defined on the accuracy event itself.

\begin{tbox}
    \begin{lemma}[Accuracy implies marginal-invertibility] \label{Lem:AutoInvertible}
        Let $\rho \in \cab$ and $\hr \in \dnab$ with $\rf(\rho, \hr) \geq 1 - \gamma$, and $\smin := \sqrt{d_\sfa \lmin(\hr_\sfa)}$.
        Then
        \begin{equation} \label{Eq:AutoInvertible}
            \smin \geq 1 - \sqrt{2 d_\sfa \gamma} .
        \end{equation}
        Consequently, $\hr_\sfa$ is invertible whenever $\gamma < 1/(2 d_\sfa)$. 
    \end{lemma}
\end{tbox}
\begin{proof}
    Observe that data-processing inequality and the assumptions in the statement of the lemma imply
    \begin{equation}
    1 - \gamma \leq \rf(\rho, \hr) \leq \rf(\rho_\sfa, \hr_\sfa) =  \rf(I_\sfa/d_\sfa, \hr_\sfa).
    \end{equation}
    Now let us define $s_i := \sqrt{d_\sfa \lambda_i(\hr_\sfa)}$, so that $\sum_i s_i^2 = d_\sfa$, which implies $\rf(I_\sfa/d_\sfa, \hr_\sfa) = d_\sfa^{-1/2} \tr \hr_\sfa^{1/2} = d_\sfa^{-1} \sum_i s_i$, whence
    \begin{equation}
        \sum_{i \in [d_\sfa]} (s_i - 1)^2 = \sum_{i \in [d_\sfa]} s_i^2 - 2 \sum_{i \in [d_\sfa]} s_i + d_\sfa = 2 d_\sfa \bigl(1 - \rf(I_\sfa/d_\sfa, \hr_\sfa)\bigr) \leq 2 d_\sfa \gamma .
    \end{equation}
    Since $\smin \equiv \min_i s_i \leq 1$, dropping every term in the above sum except this term yields $(1 - \smin)^2 \leq 2 d_\sfa \gamma$, which is~\cref{Eq:AutoInvertible}.
\end{proof}

\paragraph{Extension for a singular marginal.}
For a singular $\hr_\sfa$, let $P$ be the projector onto the support of $R = d_\sfa \hr_\sfa^{\mathsf T}$ (\cref{Eq:MarginalIsDefect}), that is, onto $\operatorname{supp}(\hr_\sfa^{\mathsf T})$.
Taking the inverse in~\cref{Eq:ClosedFormProj} on the support of $\hr_\sfa$ returns a CP map with Kraus operators $\{K_i\}$ obeying $\sum_i K_i^\dagger K_i = P \leq I_\sfa$.
The simplest repair, to obtain a map that is TP on $\sfa$, is to append a single Kraus operator, the projector onto the kernel, $K_0 := I_\sfa - P$, which restores $\sum_i K_i^\dagger K_i + K_0^\dagger K_0 = I_\sfa$ and makes~\cref{Alg:FPLSNonAdaptive,alg:FPLSadaptive} total.
As written, this requires $d_\sfa = d_\sfb$. 
For $d_\sfa \neq d_\sfb$ one takes $K_0 := V (I_\sfa - P)$ with $V$ an isometry from $\ker(P)$ into $\sfb$, which exists only if $\dim \ker(P) \leq d_\sfb$.
In general, one may instead append the operators $\ketbra{\psi}{e_k}$, for an orthonormal basis $\{\ket{e_k}\}$ of $\ker(P)$ and a fixed unit vector $\ket{\psi} \in \sfb$, which works in every dimension.
The Choi state of the extended channel differs from that of the unextended map by a term supported on the kernel block of the $\sfa$ marginal, orthogonal to the support of $\hr$, so the Choi fidelity with $\hr$ is unchanged.
Clearly, such extensions are not unique.

\subsection{The lifting loss in practice} \label{App:ContractionNumerics}
\Cref{Thm:Main} bounds the loss of the fidelity projection by a factor $2$ in Bures distance, $\db(\rho, \pc(\hr)) \leq 2\,\db(\rho, \hr)$, for every density estimate $\hr$.
\Cref{Fig:Contraction} measures the ratio $\db(\rho, \pc(\hr)) / \db(\rho, \hr)$ on tomography estimates.
Since $\db^2 = 2(1 - \rf)$ on states, the corresponding ratio of infidelities is its square, so the bound $2$ is the factor $4$ in infidelity of~\cref{Eq:ProjLoss}, and a Bures ratio $c$ is an infidelity ratio $c^2$.
For each of eight families and each $\dimab = 2^2, \dots, 2^{10}$ ($d_\sfb = d_\sfa$ for even powers of $2$ and $d_\sfb = 2 d_\sfa$ for odd ones), we draw $10$ random channels, form the density estimate of~\cref{Alg:FPLSNonAdaptive} at $\tau = \beta_N/2$ from $N$ local Pauli shots, and compare its distance to the true Choi state with that of its fidelity projection.
The shot count $N$ is the sample size of~\cref{Eq:GutaQST} at $\eps = 0.5$ and $\delta = 0.05$ with the true Choi rank (for the noisy isometries, the rank $1$ of the dominant part), which, the bound being loose, gives density estimates of infidelity between $2 \times 10^{-4}$ and $3 \times 10^{-2}$.
The families are random channels of Choi rank $1$ (Haar isometries), $2$, $\log_2 \dimab$ and $\dimab$ (BCSZ~\cite{bruzda2009random}); independent amplitude damping on $1$ to $\log_2 d_\sfa$ input qubits; Haar isometries followed by $1\%$ and $5\%$ depolarizing noise, of full Choi rank with one dominant eigenvalue; and equal mixtures of about $\sqrt{\dimab}$ Haar isometries, with a nearly flat spectrum.

\begin{figure}[!t]
    \centering
    \includegraphics[width=0.8\textwidth]{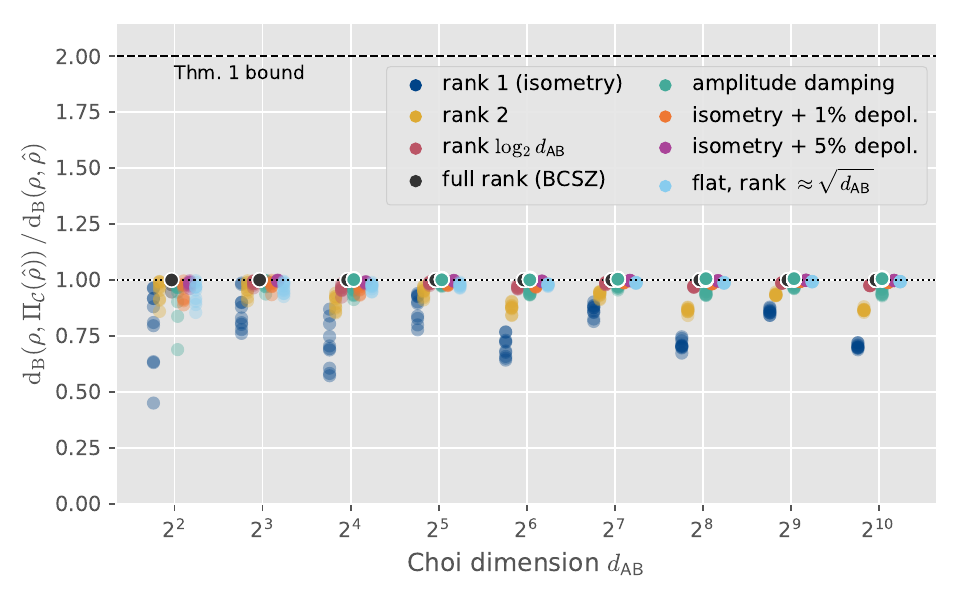}
    \caption{Ratio of the Bures distance of the channel estimate to that of the density estimate from the true Choi state, $\db(\rho, \pc(\hr))/\db(\rho, \hr)$, against the Choi dimension, for $10$ random channels per family and dimension; the families are offset horizontally for legibility. 
    Instances on which the projection is contractive (ratio at most $1$) are translucent, those on which it is not are opaque and outlined in white. The dashed line is the worst case $2$ of~\cref{Thm:Main}; squaring the ratio gives the ratio of infidelities.}
    \label{Fig:Contraction}
\end{figure}

Over all $870$ instances, the ratio lies between $0.45$ and $1.01$, with median $0.99$, against the worst case of $2$ (infidelity ratios between $0.20$ and $1.02$, against $4$).
It exceeds $1$, by at most $1\%$, only for full-rank channels (by at most $0.03\%$) and for amplitude damping whose density estimate is still below its onset and misses true eigenvalues.
For channels of low Choi rank, the projection \emph{reduces} the distance, to about $0.8$ of its value for isometries (about $0.64$ in infidelity).
The worst case of~\cref{Eq:ProjLoss} needs a density estimate far from $\cab$; a tomography estimate is close to the true Choi state, hence to $\cab$ as well, and its marginal is close to $I_\sfa/d_\sfa$.

\section{Implementation of \cref{Fig:Headline}} \label{App:HeadlineImplementation}
This appendix records, in full, how the two halves of~\cref{Fig:Headline} are produced.
The left panel isolates the \emph{cost} of TP regularization and therefore holds the quality of its input fixed while the dimension grows, whereas the right panels follow one fixed channel as the shot count grows.

\subsection{Left panel: Time-cost of TP-regularization methods} \label{App:HeadlineLadder}

\paragraph{The instance.}
At each dimension the true channel is a Haar-random isometry, of Choi rank $1$, drawn once per dimension.
Its Choi state is measured in the local Pauli basis, the least-squares estimate is formed, and the density estimate is~\cref{Alg:SparseDensityEstimate} at the PLS-QPT rule $\tau = -\lmin(\hrl)$.
This density estimate is Choi-isomorphic to a CP but not TP map.  
All three methods regularize this estimate, so the comparison is of the regularizers alone and not of the pipelines feeding them.

\paragraph{Calibration.}
As for~\cref{Fig:HIPPerformanceAnalysis} (\cref{App:SyntheticInstances}), the instances are matched in their distance to $\cab$ rather than in shot count: at every dimension the shot count is bisected in $\log N$ until $\dpu(\hr, \cab) \approx 0.05$, with the realized values listed below.
The calibrated shot counts and the realized distances are

\begin{center}
    \footnotesize
    \setlength{\tabcolsep}{3pt}
    \begin{tabular}{lccccccccc}
        \hline
        $(d_\sfa, d_\sfb)$ & $(2,2)$ & $(2,4)$ & $(4,4)$ & $(4,8)$ & $(8,8)$ & $(8,16)$ & $(16,16)$ & $(16,32)$ & $(32,32)$ \\
        $\dimab$ & $4$ & $8$ & $16$ & $32$ & $64$ & $128$ & $256$ & $512$ & $1024$ \\
        \hline
        median $N$ & $6.1{\times}10^{2}$ & $8.6{\times}10^{2}$ & $7.8{\times}10^{3}$ & $6.7{\times}10^{3}$ & $5.6{\times}10^{4}$ & $6.7{\times}10^{4}$ & $3.6{\times}10^{5}$ & $4.6{\times}10^{5}$ & $2.6{\times}10^{6}$ \\
        realized $\dpu$ & $0.055$ & $0.052$ & $0.050$ & $0.050$ & $0.051$ & $0.050$ & $0.051$ & $0.051$ & $0.050$ \\
        \hline
    \end{tabular}
\end{center}
\noindent
Ten instances are drawn per dimension, differing only in the measurement randomness.
The odd powers of two are rectangular, $d_\sfa \neq d_\sfb$; the HIP implementation of~\cite{kahn2021hip} assumes $d_\sfa = d_\sfb$ in its projection onto the TP hyperplane, and we generalize that one routine, verifying that it agrees with theirs to machine precision wherever theirs is defined, up to the ordering of the two systems, which their code takes opposite to ours.

\paragraph{The timing.}
Single-threaded wall-clock of one TP regularization, excluding the measurement and least-squares stages, which are shared.
All reported times are measured on an Apple M4 processor (10 cores, 16\,GB of memory), with Python~3.14, NumPy~2.5 and SciPy~1.18 linked against Apple's Accelerate BLAS/LAPACK, and with every computation restricted to a single thread. 
Although absolute times depend on the machine, the comparison rests on the ratios between methods and their scaling.
Our curve is the rank-factorized form of~\cref{Thm:FactorizedImplementation}, which never forms a $\dimab \times \dimab$ matrix.
HIP is run with its authors' code~\cite{kahn2021hip}, its source files unmodified, and with the default settings of its simulation driver.

The switching rules are \texttt{first} and \texttt{cos}, with $\texttt{min\_cos} = 0.99$, $\texttt{alt\_steps} = 4$, $\texttt{HIP\_steps} = 10$, $\texttt{missing\_w} = 3$, $\texttt{min\_part} = 0.1$ and at most $30$ stored hyperplanes.
The iteration is stopped once the least eigenvalue of the iterate is above $-\epsilon_{\mathrm{tol}}/\dimab$, with $\epsilon_{\mathrm{tol}} = \max\bigl(10^{-1} |\lmin(\hrl)|, 10^{-3}\bigr)$, or after $300$ iterations, and is followed by a final mixing step~\cite[Eq.~(15)]{surawy2022projected}, so that every HIP output is a valid channel.
Our only changes are the following: the projection onto the TP hyperplane is written for our ordering of the two systems and for $d_\sfa \neq d_\sfb$, as described above; one array routine that current NumPy no longer provides is replaced by its successor; and the logging of intermediate results is disabled.
The time we report is the one the code itself accounts for, which excludes that logging.
These settings are used wherever HIP's cost is reported (\cref{Fig:Headline}, left, \cref{Fig:HIPPerformanceAnalysis} and \cref{Tab:HeadToHead}), and every such run stopped on the tolerance, none on the iteration cap.
In~\cref{Fig:Headline}~(left) and~\cref{Fig:HIPPerformanceAnalysis}, HIP regularizes the calibrated density estimate described above.

In~\cref{Tab:HeadToHead} and in the accuracy figures, PLS-QPT is the authors' complete pipeline applied to the least-squares estimate, including the density-estimation routine their code calls, which differs from~\cref{Alg:SparseDensityEstimate}: the largest eigenvalues keep their least-squares values and the remaining surviving ones are lowered by $\tau$, where~\cref{Alg:SparseDensityEstimate} shifts every surviving eigenvalue by a common amount.
Wherever the accuracy or rank of PLS-QPT is reported (\cref{Fig:Headline}, right, and \cref{Fig:SpectrumMatching,Fig:AccuracyVsN,Fig:ChannelRobustness,Fig:FullRankTradeoff}), two settings depart from the defaults: the tolerance is fixed at $\epsilon_{\mathrm{tol}} = 10^{-8}$ and the iteration cap is raised to $2000$.
The reason is that the default tolerance has a floor, $10^{-3}$, that does not shrink with $N$ and would put a plateau on the infidelity of PLS-QPT: on the rank-$2$ channel of~\cref{Tab:HeadToHead} at $\dimab = 2^6$ and $N = 10^{10}$, the default settings return infidelity $2.0 \times 10^{-4}$ against $3.9 \times 10^{-5}$ at the fixed tolerance.
In~\cref{Fig:SpectrumMatching}, HIP is applied to the same density estimate as the fidelity projection.
Tightening the tolerance raises the runtime, as it does for any numerical approximation scheme: on the same channel at $N = 10^{8}$, HIP stops after $4$ iterations at the default tolerance and after $35$ at the fixed one, each iteration costing one eigendecomposition of size $\dimab$.
The costs we report for HIP are therefore those of its authors' own accuracy target, the most favourable to it among the settings we use.
The fidelity projection is a closed form, exact up to machine precision, with no tolerance to choose and a cost that does not depend on the accuracy demanded.

\paragraph{The semidefinite program.}
The diamond projection is the diamond-norm semidefinite program of~\textcite{watrous2012simpler}, applied to the difference $\fc\iv(\si) - \fc\iv(\hr)$ between a candidate channel and the CP map whose Choi state is the density estimate, with $\si$, a variable, constrained to $\cab$ by $\si \geq 0$ and $\tr_\sfb[\si] = I_\sfa/d_\sfa$.
The implementation, together with its verification against three closed-form cases, is in the repository accompanying this article~\cite{afham2026fplscode}.

\paragraph{Solver settings and stopping.}
The cost of diamond projection is decided by a tolerance we choose rather than by the algorithm's own convergence criterion, so a careless choice of stopping condition would measure our tolerance policy instead of the problem.
We therefore attempt four settings per instance, the defaults of \textsc{Clarabel} and \textsc{SCS} plus one alternative each, modelled in \texttt{cvxpy}.
The two alternatives move in opposite directions and for different reasons: \textsc{Clarabel} is an interior-point method that can stall just short of its default target, so its alternative \emph{relaxes} the gap and feasibility tolerances by one decade to $10^{-7}$; \textsc{SCS} is first order and its iterates converge slowly at the tail, so its alternative \emph{tightens} them to $\eps_{\mathrm{abs}} = \eps_{\mathrm{rel}} = 10^{-9}$ with a raised iteration cap.
A run counts only if its output is a valid channel, with marginal error below $10^{-6}$ and least eigenvalue above $-10^{-6}$; among the valid runs the fastest is plotted.

\paragraph{SDP timings.}
Each instance is solved three times with warm starting disabled and the minimum is reported.
The number plotted is the solver's own reported solve time, which excludes \texttt{cvxpy} modelling and canonicalization; those are one-off costs of the interface rather than of the projection.
All solvers run single-threaded, as do HIP and the fidelity projection, so the three curves are comparable.
The largest instances carry a ten-minute cap per solve, passed to each solver as its own time limit.

\subsection{Right panels: rank and accuracy against the shot count} \label{App:HeadlineRight}
The true channel is fixed: the mixed-unitary channel on four qubits that applies the quantum Fourier transform with probability $1/2$ and a fixed Haar-random unitary otherwise, so $\dimab = 2^8$, the Choi rank is $2$ and both nonzero Choi eigenvalues are close to $1/2$.
At each shot count, ten independent measurement records are drawn; from each, the (same) least-squares estimate is formed and handed to both FPLS- and PLS-QPT algorithms.
The two pipelines are those of~\cref{Tab:PLSvsFPLS}, with the density estimate of PLS-QPT formed by the routine of its authors' code, which differs from~\cref{Alg:SparseDensityEstimate} as described in~\cref{App:HeadlineLadder}.
Ranks count eigenvalues above $\eps_{\mathrm{zero}} = 10^{-10}$; rank curves are the mode over the ten trials, ties broken upward, and infidelity curves are the geometric mean over the ten trials.
The dotted guides carry the \emph{predicted} exponents $1/N$ and $1/\sqrt N$, which are never fitted. 
The guides are displaced off the curve they describe for visual distinguishability.

\end{document}